\documentclass[12pt, draftclsnofoot, onecolumn]{IEEEtran}
\usepackage{amssymb}
\usepackage{amsmath}
\usepackage{cite}
\usepackage{url}
\usepackage{xcolor}
\usepackage{cite,graphicx,amsmath,amssymb}
\usepackage{subfigure}
\usepackage{fancyhdr}
\usepackage{mdwmath}
\usepackage{mdwtab}
\usepackage{caption}
\usepackage{amsthm}
\usepackage{setspace}
\usepackage{algorithm}
\usepackage{algorithmic}
\usepackage{makecell}
\usepackage{diagbox}
\usepackage{multirow}

\newtheorem{remark}{Remark}
\newtheorem{theorem}{Theorem}

\newtheorem{lemma}{Lemma}

\newtheorem{corollary}{Corollary}

\newtheorem{proposition}{Proposition}
\allowdisplaybreaks
\makeatletter
\def\ScaleIfNeeded{%
\ifdim\Gin@nat@width>\linewidth \linewidth \else \Gin@nat@width
\fi } \makeatother

\begin{document}

\title{Two-Timescale Design for STAR-RIS Aided NOMA Systems}

\author{
	Chenyu~Wu, Changsheng~You, \IEEEmembership{Member, IEEE}, Yuanwei~Liu, \IEEEmembership{Senior Member,~IEEE}, Shuo Shi, \IEEEmembership{Member,~IEEE}, and Marco Di Renzo, \IEEEmembership{Fellow,~IEEE} 	\vspace{-2em}
	

		\thanks{C. Wu, and S. Shi are with the School of Electronics and Information Engineering, Harbin Institute of Technology (HIT), Harbin 150001, China (e-mail: \{wuchenyu, guxuemai\}@hit.edu.cn).}
		\thanks{C. You is with the Department of Electronic and Electrical Engineering, Southern University of Science and Technology (SUSTech), Shenzhen 518055, China (email: youcs@sustech.edu.cn). 
		}
		\thanks{Y. Liu is with the School of Electronic Engineering and Computer Science, Queen Mary University of London, London E1 4NS, UK (email: yuanwei.liu@qmul.ac.uk).}
		\thanks{M. Di Renzo is with Universit\'{e} Paris-Saclay, CNRS, CentraleSup\'{e}lec, Laboratoire des Signaux et Syst\`{e}mes, 3 Rue Joliot-Curie, 91192 Gif-sur-Yvette, France. (email: marco.di-renzo@universite-paris-saclay.fr).}

	}
\maketitle
\vspace{-1em}
\begin{abstract}
Simultaneously transmitting and reflecting reconfigurable intelligent surfaces (STAR-RISs) have emerged as a promising technology to reconfigure the radio propagation environment in the full space. Prior works on STAR-RISs have mostly considered the energy splitting operation protocol, which has high hardware complexity in practice. Moreover,  
the full and instantaneous channel state information (CSI) is always assumed available for designing the STAR-RIS passive beamforming, which, however, is practically difficult to obtain due to the large number of STAR-RIS elements.
To address these issues, we study the \emph{mode switching} design in  STAR-RIS aided non-orthogonal multiple access (NOMA) communication systems. Moreover, two efficient \emph{two-timescale} (TTS) transmission protocols are proposed for different channel setups to maximize the respective average achievable sum-rate. Specifically, 1) for the case of line-of-sight (LoS) dominant channels, we propose the \emph{beamforming-then-estimate} (BTE) protocol, where the long-term STAR-RIS transmission and reflection coefficients are optimized based on the statistical CSI only, while the short-term power allocation at the base station (BS) is designed based on the estimated effective fading channels of all users; 2) for the rich scattering environment, we propose an alternative \emph{partition-then-estimate} (PTE) protocol, where the BS first determines the long-term STAR-RIS surface-partition strategy based on the path-loss information only, with each subsurface being assigned to one user; and then the BS estimates the instantaneous subsurface channels associated with the users and designs its power allocation and STAR-RIS phase-shifts accordingly. For the two proposed transmission protocols, we further propose efficient algorithms to solve the respective long-term and short-term optimization problems. Moreover, we show that both proposed transmission protocols substantially reduce the channel estimation overhead as compared to the existing schemes based on full instantaneous CSI. Last, simulation results validate the superiority of our proposed transmission protocols as compared to various benchmarks.  It is shown that the BTE protocol outperforms the PTE protocol when the number of STAR-RIS elements is large and/or the LoS channel components are dominant, and vice versa.
\end{abstract}
\vspace{-1em}
\begin{IEEEkeywords}
\vspace{-1em}	
Reconfigurable intelligent surface (RIS), simultaneous transmission and reflection (STAR), non-orthogonal multiple access (NOMA), statistical CSI, two-timescale optimization.
\end{IEEEkeywords}

\section{Introduction}
Reconfigurable intelligent surface (RIS)\cite{ris_survey2,survey_ris}, also known as intelligent reflecting
surface (IRS)\cite{zr_survey}, has been envisioned as a promising technology to enhance the spectral efficiency of wireless communication systems. Specifically, RIS is a planar array comprising a massive number of tunable and low-cost elements. By smartly controlling the amplitude/phase-shift of individual elements, RIS can program the wireless propagation environment for achieving various functions, such as enhancing the desirable signal strength\cite{irs_joint} or suppressing the interference\cite{mimoris}. Nevertheless, the  conventional reflecting-type RISs  can only achieve half-space reflection as specified by their pointing directions, thereby limiting the network coverage. 
To tackle this issue, a new type of RIS, called \emph{simultaneous transmitting and reflecting reconfigurable intelligent surfaces} (STAR-RISs)\cite{STAR}, or intelligent omni-surface (IOS)\cite{ios}, has been recently proposed. In contrast to conventional RISs\cite{ris_survey2,survey_ris,zr_survey,irs_joint,mimoris}, STAR-RISs enable simultaneously reflecting and/or transmitting the
signals incident on both sides of the surface, thus achieving $\text{360}^\circ$ wireless coverage. Particularly, 
the amplitude and phase responses of the transmitted and reflected signals can be jointly and appropriately optimized, which provides an extra degree of freedom for enhancing the communication performance.


On the other hand, non-orthogonal multiple access (NOMA) is a promising technology to satisfy the stringent requirements of next-generation wireless networks such as extremely high spectral efficiency and massive connectivity\cite{proceeding}. Different from the conventional orthogonal multiple-access (OMA) that serves users over orthogonal time/frequency resource blocks (RBs), NOMA allows multiple users to share the same RB in the code or power domain. Specifically, for power-domain NOMA\cite{power_noma2}, the users’ signals are superposed at the transmitter and recovered at the receiver by using the successive interference cancellation (SIC) technique. Existing works have revealed that NOMA is able to achieve a much higher rate than OMA when the users' channels are more distinctive\cite{ding_noma_tvt}. This thus motivates the study of STAR-RIS aided NOMA systems, as STAR-RISs can flexibly manipulate the wireless channels to achieve enhanced NOMA gain.

\vspace{-1em}
\subsection{Related Works}

\emph{1) RIS-aided Wireless Systems:} In the existing literature, RIS has been extensively investigated in various setups, such as orthogonal frequency division multiplexing (OFDM)\cite{ofdm}, millimeter wave (mmWave) communications\cite{irs_mmwave}, unmanned aerial vehicle (UAV) communications \cite{huameng_uav}, NOMA transmissions\cite{risnoma4,risnoma7,risnoma8}, and so on. These works have mostly assumed the availability of instantaneous channel state information (CSI) for reaping the RIS passive beamforming gain. However, obtaining the full and accurate CSI is practically difficult due to the large number of reflecting elements and their passive nature. As a result, various approaches have been proposed to efficiently estimate the RIS-reflected channels. For example, under the single-user setup, the 
authors of \cite{dft} and \cite{discrete_jsac} designed efficient RIS reflection training patterns to estimate the cascaded channels under the ideal case with continuous phase-shifts and the practical case with discrete phase-shifts, respectively. Particularly, it was shown that the length of the required training sequence for channel estimation scales with the number of RIS elements\cite{dft}. For the multiuser scenario, various channel properties such as the common base station (BS) to RIS channel \cite{ce_ofdma}, channel spatial correlation \cite{ll}, and channel sparsity\cite{part2} were leveraged for reducing the channel estimation overhead. A fast multi-beam training
method for RIS-assisted mmWave system was proposed in \cite{beamtraining}. Moreover, since the statistical CSI (S-CSI) varies slowly and can hence be more easily obtained, efficient two-timescale (TTS) RIS beamforming designs have been proposed that exploit the S-CSI for designing RIS reflections and the instantaneous CSI (I-CSI) for designing the BS transmit beamforming\cite{scsi1,scsi2,downlink_scsi,kangda2,marco3}. Despite the prominent passive beamforming gain, the performance of wireless systems aided by an RIS are fundamentally constrained by its serving region. This thus motivates the design of a new STAR-RIS that is capable of serving the users on both of its sides. 

\emph{2) STAR-RIS-assisted Wireless Systems:} Recently, growing research attention has been devoted to investigating the benefits of deploying STAR-RISs in wireless communication systems. Specifically, the authors of \cite{SRAR} presented an hardware model for STAR-RIS and its channel models in the near-field and far-field regions. In \cite{ios2}, a circuit-based reflection-refraction model for the IOS was proposed and verified by experimental results.
Besides, in \cite{STAR,star1}, three STAR-RIS operation protocols have been proposed, namely \emph{energy splitting (ES), mode switching (MS)}, and \emph{time switching (TS)}. Based on these protocols, the joint active and passive beamforming design was studied to minimize the total power consumption at the access point\cite{star1}. For multiuser radio access, the integration of STAR-RISs and NOMA has drawn great attention. For example, the authors of \cite{star2} revealed that the coverage range can be extended when employing STAR-RISs in NOMA networks. Besides, a STAR-RIS-assisted multi-cluster NOMA system was studied in \cite{star4}, where the decoding order, power allocation, active beamforming, and STAR-RIS coefficients were jointly optimized for maximizing the achievable sum-rate. 
Moreover, the authors of \cite{analysis_ios} proposed three channel models for a STAR-RIS aided NOMA network, based on which the diversity gain was analyzed.

\vspace{-1em}
\subsection{Motivation and Contributions}

In view of the existing works on STAR-RIS aided systems, there are two  limitations that have not been well addressed. First, most works have considered the ES operation protocol for STAR-RIS. However, this protocol generally incurs a relatively high hardware complexity, since it necessitates the simultaneous manipulation of both the transmitted and reflected signals for each STAR-RIS element. Moreover, the ES protocol requires the joint optimization of the energy splitting ratio and transmission/reflection phase-shifts, which is generally difficult to solve. Second, existing works on STAR-RIS 
mainly focused on its passive beamforming design for data transmission, under the assumption of full I-CSI available at the transmitter. Nevertheless, as recently reported in \cite{ce_survey,star_ce}, acquiring the I-CSI associated with all STAR-RIS elements and all users usually incurs a prohibitively high channel estimation overhead, which, in general, is proportional to the number of users and STAR-RIS elements. In practical systems, a large channel training overhead may result in less time remaining for data transmission and may hence degrade the communication performance. This issue becomes more severe when the wireless channels are highly dynamic due to the short channel coherence time.

To overcome these issues, we study the MS design for a STAR-RIS aided NOMA system. Compared with ES, MS is much easier to implement, since each STAR-RIS element can be flexibly tuned to operate in either the transmission or reflection mode, thus avoiding the intrinsic coupling between them. Moreover, to reduce the channel estimation overhead for acquiring instantaneous and full CSI, we propose two efficient TTS transmission protocols for two channel setups. The main contributions of this paper are summarized as follows.

\begin{itemize}
	\item  First, we consider the channel setup where the STAR-RIS associated channels are line-of-sight (LoS) dominant. For this case, we propose an efficient TTS transmission protocol called \emph{beamforming-then-estimate} (BTE). Specifically, the BS exploits the S-CSI for designing the STAR-RIS beamforming coefficients (including the mode selection and phase-shift of each element) in the long term and keeps them fixed as long as the S-CSI remains unchanged. Meanwhile, in each time slot, the BS estimates the instantaneous effective CSI for each user and designs its short-term power allocation accordingly. To maximize the average sum-rate, we formulate an optimization problem 
	 that jointly optimizes the long-term STAR-RIS beamforming coefficients and the short-term BS power allocation. To tackle this challenging problem, we first solve the short-term optimization problem to obtain the BS power allocation based on the effective CSI. Then,  an efficient penalty-based method is developed to suboptimally solve the long-term optimization problem for obtaining the STAR-RIS transmission- and reflection-coefficient vectors based on the S-CSI.  Particularly, for the BTE protocol, the channel estimation overhead  is equal to the number of users $K$.
	\item  Next, as the BTE protocol is ineffective in rich scattering environments, we propose another dedicated TTS transmission protocol for this channel setup, which is called \emph{partition-then-estimate} (PTE). Specifically, based on the information of large-scale path-loss, the STAR-RIS is partitioned into several transmission/reflection subsurfaces, each associated with one specific user. 
	Then, in each time slot, the BS estimates the instantaneous subsurface channels associated with each user (referred to as the partial I-CSI), based on which it jointly designs the short-term BS power allocation and STAR-RIS phase-shifts.
	Similarly, the short-term optimization problem is optimally solved based on the partial I-CSI. For the long-term optimization problem, on the other hand, we propose a low-complexity yet efficient algorithm to determine the proper number of elements assigned to each user. The channel estimation overhead in each time slot for this protocol is $K+M$, where $M$ denotes the number of STAR-RIS elements.
	\item Last, simulation results are presented under various system setups to demonstrate the effectiveness of our proposed TTS transmission protocols as compared to several benchmarks. 
	It is shown that the BTE protocol is more effective when $M$ is large and the LoS channel components are dominant, whereas the PTE protocol is preferred when $M$ is small and/or the channel coherence time is sufficiently long.
\end{itemize}

\vspace{-1em}
\subsection{Organization and Notations}

The rest of this paper is organized as follows. In Section II, we present the system model for the considered STAR-RIS aided NOMA communication system. In Sections III and IV, we propose two practical TTS transmission protocols for LoS-dominant channels and rich scattering channels, respectively. Efficient algorithms for solving the corresponding average sum-rate maximization problem are also provided. Section V presents numerical
results to validate the effectiveness of the proposed designs. Finally, Section VI concludes this paper.

\emph{Notations:}
 $\mathbb{C}^{M\times1}$ denotes the space of $M\times1$ complex-valued
vectors. $\mathbf{a}^H$ denotes the conjugate transpose of vector $\mathbf{a}$. $\text{diag}(\mathbf{a})$
denotes a diagonal matrix with the diagonal elements given by $\mathbf{a}$. $\text{diag}(\mathbf{A})$
denotes a vector extracted from the diagonal elements of matrix $\mathbf{A}$. The distribution
of a circularly symmetric complex Gaussian variable (CSCG) with mean $\mu$ and variance $\sigma^2$ is denoted by $\mathcal{C}\mathcal{N}\sim{(\mu,\sigma^2)}$.
The $m$-th element of a vector $\mathbf{a}$ is denoted by $[\mathbf{a}]_m$.  For a complex scalar $x$, $|x|$, $\angle{x}$ and $\text{Re}(x)$ denote its modulus, angle and real part, respectively. $\triangleq$ stands for definition.

\vspace{-1em}
\section{System Model}

\begin{figure}[t!]
	\centering
	\includegraphics[width=0.6\textwidth]{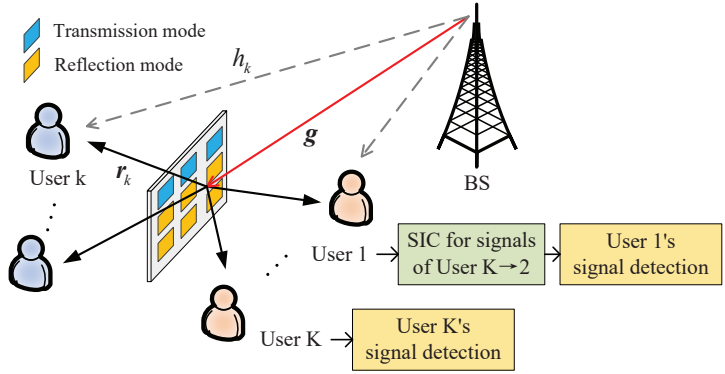}\\
	\caption{A STAR-RIS aided multiuser NOMA communication system, where the STAR-RIS operates under the mode switching protocol.}\vspace{1em}
\end{figure}

As depicted in Fig. 1, we consider a narrow-band wireless communication system, where a STAR-RIS comprising $M$ elements is deployed to assist the downlink communication from a single-antenna BS to $K$ single-antenna users, denoted by the set $\mathcal{K}=\{1,...,K\}$. The STAR-RIS divides the full space of signal propagation into two regions, namely, the transmission region and the reflection region. Accordingly, the users are divided into two groups, namely T users and R users, which are denoted by $\mathcal{K}_t=\{1,...,K_t\}$ and $\mathcal{K}_r=\{1,...,K_r\}$, respectively.

Let $h_{k}\in\mathbb{C}$ denote the baseband equivalent channel from the BS to user $k\in\mathcal{K}_s$, with $s\in\{t,r\}$ specifying the transmission/reflection region where the user is located. Further, let $\mathbf{r}_{k}^H\in\mathbb{C}^{1\times{M}}$ and $\mathbf{g}\in\mathbb{C}^{M\times{1}}$ denote the channel from the STAR-RIS to user $k$ and from the BS to the STAR-RIS, respectively. We assume the Rician fading channel model for the links associated with the STAR-RIS, i.e., \{$\mathbf{r}_{k}^H$\} and $\mathbf{g}$;  while the BS-user links \{$h_{k}$\} are modeled by Rayleigh fading. As such, the channel coefficients are given by
\begin{subequations}
	\begin{align}
		&{\mathbf{g}} = \sqrt {\delta_{BS}} \left( {\sqrt {\frac{{{\kappa_1}}}{{{\kappa_1} + 1}}} {\bar{\mathbf{g}}} + \sqrt {\frac{1}{{{\kappa_1} + 1}}} {\hat{\mathbf{g}}}} \right),\\
		&{{\mathbf{r}}_{k}} = \sqrt {\delta_{Sk}} \left( {\sqrt {\frac{{{\kappa_2}}}{{{\kappa_2} + 1}}} {\bar{\mathbf{r}}_{k}} + \sqrt {\frac{1}{{{\kappa_2} + 1}}} {\hat{\mathbf{r}}_{k}}} \right),\forall{k}\in\mathcal{K},\\
		&h_{k}=\sqrt {\delta_{k}}\hat{h}_k,\forall{k}\in\mathcal{K},
	\end{align}
\end{subequations}
where $\delta_{BS}$, $\delta_{Sk}$, $\delta_{k}$ are the distance-dependent path-loss of the BS to STAR-RIS, STAR-RIS to user $k$, BS to user $k$ links, respectively, $\bar{\mathbf{g}}$ and $\bar{\mathbf{r}}_k$ denote the normalized LoS components,  $\hat{\mathbf{g}}$, $\hat{\mathbf{r}}_k$ and $\hat{h}_k$ denote the normalized non-LoS (NLoS) components, which follow the CSCG distribution with zero mean and unit variance, and $\kappa_1,\kappa_2$ are the Rician factors. 


\vspace{-1em}
\subsection{STAR-RIS Operation Protocol: Mode Switching}
We consider the mode switching operation protocol for the studied STAR-RIS aided communication system. According to \cite{star1}, the STAR-RIS coefficients can be characterized by the transmission and reflection matrices $\boldsymbol{\Theta}_s=\text{diag}(\boldsymbol{\theta}_s)\text{diag}(\boldsymbol{\beta}_s)$, where $\boldsymbol{\theta}_{s}= [{e^{j{\theta_1^s}}},{{e^{j{\theta_2^s}}}}, \ldots ,{e^{j{\theta _M^s}}}]^T$, $\boldsymbol{\beta}_s=[{\sqrt{\beta_1^s}},...,{\sqrt{\beta_M^s}}]^T,s\in\{t,r\}$, with $\beta_m^s\in\{0,1\}$ and $\theta_m^s\in[0,2\pi),m\in\mathcal{M}=\{1,...,M\}$ denoting the operation mode and phase-shift of the $m$-th element, respectively. Further, we define $\mathbf{v}_s\triangleq{\text{diag}}(\boldsymbol{\Theta}_s), s\in\{t,r\}$ as the transmission- and reflection-coefficient vectors. 
According to the law of energy conservation, we have $\beta_m^t+\beta_m^r\leq1$ or equivalently $|[\mathbf{v}_t]_m|^2+|[\mathbf{v}_r]_m|^2\leq1$ needs to be satisfied. 
Therefore, each element of STAR-RIS operates in either the full transmission or full reflection mode, depending on the value of $\beta_m^s$. As such, the STAR-RIS-assisted effective channel  between 
the BS and user $k$ can be expressed as
\begin{equation}\label{effective_channel}
	c_k(\mathbf{v}_s)={h}_{k}+\mathbf{r}_{k}^H\boldsymbol{\Theta}_s\mathbf{g}=h_k+\mathbf{w}_k^H\mathbf{v}_s,
\end{equation}
where ${\mathbf{w}}_k^H={{\mathbf{r}}_{k}^H{\text{diag}}({\mathbf{g}})}$.

It is worth noting that the MS protocol can be regarded as a special case of the ES protocol where the amplitude coefficients for transmission and reflection are set as continuous variables, i.e., $\beta_m^s\in[0,1]$. Compared with the ES protocol that provides high STAR-RIS reflection/transmission design flexibility, the MS protocol is particularly attractive for practical use due to the following reasons. Firstly, the transmission/reflection mode switch of each element is easy to realize from the hardware implementation perspective \cite{star1}. Secondly, under the ES protocol, more STAR-RIS coefficients need to be optimized, which requires higher computational complexity than the MS protocol. Thirdly, as recently reported in\cite{couple2}, the MS protocol is not restricted by the transmission-reflection phase correlation constraint which may cause performance degradation to the ES protocol. Finally, once the transmission/reflection modes of the STAR-RIS elements are fixed, the overall channel estimation overhead for STAR-RIS MS can be reduced by half, as compared with STAR-RIS under ES.

\vspace{-1em}
\subsection{NOMA Transmission}
We consider NOMA transmission for all users. By adopting superposition coding, the received signal at user $k$ is given by
\begin{equation}
	y_{k}=c_k(\mathbf{v}_s)\sum_{k=1}^{K}\sqrt{p_k}s_k+n_{k},
\end{equation}
where $s_k\sim {\mathcal{C}\mathcal{N}}\left(0,1\right)$ denotes the transmitted data symbol for user $k$,  $n_{k}\sim {\mathcal{C}\mathcal{N}}\left(0,\sigma^2\right)$ is the additive white Gaussian noise, $p_k$ is the transmit power allocated to user $k$. For NOMA transmission, by invoking the SIC technique, the user with the stronger channel power gain is able to first decode the signals of the other users, and to then cancel their contributions before decoding its own signal. However, since the effective channel power gain of each user is determined by the STAR-RIS coefficients according to (\ref{effective_channel}), the optimal decoding order has $K!$ combinations. To reduce the computational complexity for the exhaustive search, we propose the following simple scheme for the decoding order. Without loss of generality, we assume that the users are arranged in  descending order according to their distances to the BS, i.e., $l_1>l_2>...>l_K$. Then, we assume that the NOMA decoding order also follows this permutation, i.e., user 1 is always the strongest user whose signals are last decoded. 
To guarantee the success of SIC, the effective channel power gain should satisfy $|c_k(\mathbf{v}_s)|^2>|c_j(\mathbf{v}_s)|^2, \forall{k<j}$. As such, the achievable rate of user $k$ in bits/second/Hz (bps/Hz) is given by
\begin{equation}\label{noma_expect_rate}
R_k= \log_2\left( 1+\frac{p_k{|c_k(\mathbf{v}_s)|^2}}{\sum_{j>k}p_{j}|c_j(\mathbf{v}_s)|^2+\sigma^2}\right).
\end{equation}

\subsection{Conventional Transmission Protocol: Full I-CSI based Design}

In the existing literature, the conventional transmission protocol for STAR-RIS is as follows. 
In each channel coherence block, the BS first estimates the full I-CSI of all element-wise channels (i.e., the cascaded channels $\mathbf{w}_k$) as well as the direct channels $h_k$. Then, the BS determines the STAR-RIS transmission/reflection coefficients based on the estimated full I-CSI. However, this approach entails at least $KM+K$ training symbols by using existing channel estimation methods\cite{star_ce}, which incurs prohibitively high channel estimation overhead, and hence reduces the average achievable rate due to a short time for data transmission.

To address this issue, in the next sections, we propose two practical transmission protocols for balancing the tradeoff between system performance and channel estimation overhead.

\section{TTS Transmission Protocol for LoS-Dominant Channels}

In this section, we consider the general scenario where the STAR-RIS associated channels are LoS-dominant. For this scenario, we propose an efficient TTS transmission protocol, called BET protocol; according to which we optimize the long-term STAR-RIS transmission/reflection coefficient vectors based on the S-CSI, and the short-term BS power allocation based on the  effective I-CSI. 
\vspace{-1em}
\subsection{Proposed BTE Transmission Protocol}
As shown in Fig. \ref{BTE}, the proposed BTE protocol consists of the following three main phases. In the first phase, the long-term S-CSI (including the information of the path-loss, Rician factors and the LoS components of all channels) is first estimated at the STAR-RIS and BS. Then, based on the S-CSI, the operation modes and phase-shifts of the STAR-RIS elements are designed and kept fixed for the subsequent coherence blocks for which the S-CSI remains unchanged. In the second phase, the BS estimates the instantaneous effective channels of all users in each time block, based on which the transmit power allocation is optimized for maximizing the achievable sum-rate. In the third phase, after effective CSI estimation, the BS transmits data to the users.
 \begin{figure}[t!]
 	\centering
 	\includegraphics[width=0.9\textwidth]{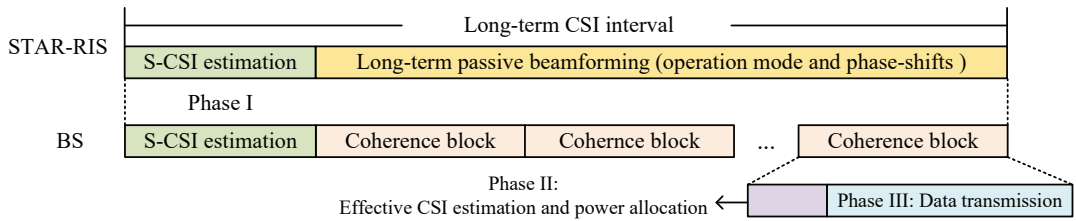}\\
 	\caption{Illustration of the proposed BTE protocol.}\vspace{1em}\label{BTE}
 \end{figure}

The benefits of this protocol are as follows. First, the S-CSI mainly depends on the locations and channel statistics which vary slowly and remain largely invariant for a number of coherence blocks. As long as the S-CSI remains unchanged, the STAR-RIS coefficients are kept fixed regardless of the instantaneous channel variations. Therefore, compared to the full I-CSI based design which requires the instantaneous configuration of the STAR-RIS, the proposed BTE protocol greatly reduces the computational complexity and feedback overhead. Second, in each coherence block, the BS only needs to estimate the effective CSI (i.e., $|c_k(\mathbf{v}_s)|^2,\forall k$), which only requires one pilot symbol for each user. Hence, the overall channel training overhead for the BTE protocol is equal to $K$, which is much smaller than that of the conventional full I-CSI based transmission protocol (i.e., $KM+K$).\footnote{Since 
the S-CSI is only estimated in a long-term interval, we neglect the associated estimation overhead.}

For simplicity, we assume that the S-CSI and the effective I-CSI are perfectly acquired at the BS/STAR-RIS via existing channel estimation methods (see e.g., \cite{mingmin,aoa}); while the extension to a robust design under imperfect CSI is left for future work.

\vspace{-1em}
\subsection{Problem Formulation}
By considering the channel estimation overhead, the achievable rate of user $k$ is given by $R_k^a=(1-\frac{K}{T_c})R_k$, where $T_c$ denotes the channel coherence time (normalized by the symbol duration). For the proposed BTE protocol, we aim to maximize the average sum-rate by optimizing the long-term STAR-RIS transmission- and reflection-coefficient vectors $\{\mathbf{v}_s\}$ and the short-term BS power allocation $\{p_k\}$, subject to constraints on the individual user’s rate given by $\gamma_k$.  Mathematically, this optimization problem can be formulated as
\begin{subequations}
	\begin{align}
		(\text{P1}):\; { \mathop{\max}\limits_{\{\mathbf{v}_s\}}}\;&\mathbb{E}\left\lbrace  {\mathop{\max}\limits_{\{p_k\}}}\; \sum_{k=1}^K{R_k^a} \right\rbrace  \\
		{\rm{s.t.}}\;\;	
		\label{P_rate}&{R}_k^a \geq \gamma_k,\forall{k\in\mathcal{K}},\\
		&|c_k(\mathbf{v}_s)|^2>|c_j(\mathbf{v}_s)|^2, \forall{k<j}, \\
		\label{P_energy}&|[\mathbf{v}_t]_m|^2+|[\mathbf{v}_r]_m|^2\leq1,\forall{m\in\mathcal{M}},\\
		\label{P_mode}&|[\mathbf{v}_s]_m|\in\{0,1\},\forall{m}\in\mathcal{M},\forall{s}\in\{t,r\},\\
		\label{P_power}&\sum_{k=1}^K{p_k}\leq{P_{\text{max}}},
	\end{align}
\end{subequations}
where the expectation in (5a) is taken over the random channel realizations.
Note that the inner rate maximization problem is to optimize the short-term BS power allocation in each time slot given the STAR-RIS coefficient vectors $\{\mathbf{v}_s\}$, while the outer rate maximization problem is to optimize $\{\mathbf{v}_s\}$. Problem (P1) is challenging to solve since the design variables are highly coupled in the objective function and constraints. Moreover, it is difficult to obtain a closed-form expression for $\mathbb{E}\left\lbrace \sum_{k=1}^K{R_k^a} \right\rbrace$. To tackle these challenges, we propose an efficient algorithm to obtain a high-quality suboptimal solution in the next subsection.

\vspace{-1em}
\subsection{Proposed Solution to Problem (P1)}

\emph{i) Short-term BS power allocation optimization:} 
First, from problem (P1), the short-term optimization problem reduces to 
\begin{subequations}
	\begin{align}
		(\text{P1.1}):\;  & {\mathop{\max}\limits_{\{p_k\}}}\; \sum_{k=1}^K{R_k^a} \\
		{\rm{s.t.}}\;\;	
		&\rm \eqref{P_rate},\eqref{P_power}.\nonumber
	\end{align}
\end{subequations}

The optimal solution to (P1.1) is given as follows.
\begin{proposition}
	\rm Given the effective channel power gains $\{|c_k(\mathbf{v}_s)|^2\}$, the optimal BS power allocation is given by
	\begin{equation}\label{optimal_power}
		\left\lbrace 
		\begin{aligned}
			&p_K^*=\frac{2^{\gamma_K}-1}{2^{\gamma_K}}\left(P_{\text{max}}+\frac{\sigma^2}{{|c_K(\mathbf{v}_s)|^2}}\right),\\
			&\vdots\\
			&p_2^*=\frac{2^{\gamma_2}-1}{2^{\gamma_2}}\left(P_{\text{max}}-\sum_{k=3}^K{p_k^*}+\frac{\sigma^2}{{|c_2(\mathbf{v}_s)|^2}}\right),\\
			&p_1^*={P_{\text{max}}}-\sum_{k=2}^K{p_k^*},
		\end{aligned}
		\right.
	\end{equation}
	and problem P(1.1) is feasible if the solution satisfies $p_1^*\geq\frac{(2^{\gamma_1}-1)\sigma^2}{|c_1(\mathbf{v}_s)|^2}$ and $p_k^*>0,\forall{k>1}$.
\end{proposition}
\begin{proof}
	See Appendix A.
\end{proof}


\emph{ii) Long-term STAR-RIS transmission- and reflection-coefficient vectors optimization:} 

For optimizing the STAR-RIS coefficients by exploiting the S-CSI, we first provide a closed-form expression for the expectation of the channel power gain $|c_k(\mathbf{v}_s)|^2$ .

\begin{lemma}
	\rm
	For the BTE protocol, the expected effective channel power gain of user $k$, i.e., $\mathbb{E}\{|c_k(\mathbf{v}_s)|^2\}$, is given by
	\begin{equation}\label{expect_channel}
		\mathbb{E}\{|c_k(\mathbf{v}_s)|^2\}=\mathbb{E}\left\lbrace |{h}_{k}+\mathbf{r}_{k}^H\boldsymbol{\Theta}_s\mathbf{g}|^2\right\rbrace =\delta_k+\epsilon_k\left|  {\bar{\mathbf{w}}_k^H\mathbf{v}_s}\right|^2+\zeta_k\sum_{m=1}^{M}\beta_m^s,
	\end{equation}
	where $\epsilon_k=\frac{\kappa_1\kappa_2{\delta_{BS}\delta_{Sk}}}{{({\kappa_1} + 1)(\kappa_2+1)}}$, $\zeta_k=\frac{(\kappa_1+\kappa_2+1){\delta_{BS}\delta_{Sk}}}{{({\kappa_1} + 1)(\kappa_2+1)}}$, and $\bar{\mathbf{w}}_k^H\triangleq{\bar{\mathbf{r}}_{k}^H{\rm{diag}}(\bar{\mathbf{g}})}\in\mathbb{C}^{1\times{M}}$ is defined as the cascaded LoS channel.
\end{lemma}

\begin{proof}
	See Appendix B.
\end{proof}

Next, we focus on the average achievable rate. Since it is difficult to obtain a closed-form expression for $\mathbb{E}\{R_k^a\}$, we approximate it in the following lemma.
\vspace{-0.5em}
\begin{lemma}
	\rm For the BTE protocol, the expected achievable communication rate of user $k$ can be approximated as
	\begin{equation}\label{noma_expect_rate_approx}
		\mathbb{E}\{R_k^a\}\approx \left(1-\frac{K}{T_c}\right)\log_2\left( 1+\frac{p_k{\mathbb{E}\{|c_k(\mathbf{v}_s)|^2\}}}{\sum_{j>k}p_{j}\mathbb{E}\{|c_j(\mathbf{v}_s)|^2\}+\sigma^2}\right)\triangleq{\bar{R}_k}, 
	\end{equation}
\end{lemma}
\vspace{-1em}
\begin{proof}
	The proof is similar to that of [Theorem 1, \cite{huameng_uav}] and hence is omitted for brevity.
\end{proof}

For the long-term optimization problem, directly substituting the optimal power allocation in \eqref{optimal_power} into (\ref{noma_expect_rate_approx}) will result in complicated rate expressions. To address this difficulty, we introduce a set of auxiliary variables $\{\bar{p}_k\}$ that denote the reference power allocation policy based on the S-CSI. Note that $\{\bar{p}_k\}$ are used for designing the STAR-RIS coefficient vectors in the long-term optimization only, while the actual BS power allocation policy in each coherence block is given by \eqref{optimal_power}.

Based on the above approximations and introduced auxiliary variables, the long-term optimization problem can be approximated as follows:
\begin{subequations}
	\begin{align}
		(\text{P2}): \mathop  {\max}\limits_{\{\mathbf{v}_s,\bar{p}_k\}}\;\; &\sum_{k=1}^K\bar{R}_k \\
		\label{P1_rate}{\rm{s.t.}}\;\;	
		&\bar{R}_k \geq \gamma_k,\forall{k\in\mathcal{K}},\\
		\label{P1_do}&\mathbb{E}\{|c_k(\mathbf{v}_s)|^2\}>\mathbb{E}\{|c_j(\mathbf{v}_s)|^2\}, \forall{k<j}, \\
		&\sum_{k=1}^K{\bar{p}_k}\leq{P_{\text{max}}},\\
		&\rm \eqref{P_energy},\eqref{P_mode}.\nonumber
	\end{align}
\end{subequations}
The main challenges for solving (P2) lie in the non-convex objective function as well as the non-convex constraints \eqref{P_mode}, \eqref{P1_rate} and \eqref{P1_do}. To tackle these difficulties, we propose an efficient two-layer penalty-based algorithm. Specifically, in the inner layer, we solve a penalized optimization problem by applying the alternating optimization (AO) method, while in the outer layer, the penalty coefficient is updated until convergence is achieved. 

\emph{1) Inner Layer Iteration:}
First, in the inner layer, the original problem (P2) is decomposed into two subproblems: BS transmit power optimization and transmission- and reflection-coefficient vectors optimization.

\textbf{Optimize} $\bar{p}_k$ \textbf{for given} $\mathbf{v}_s$:  For given $\mathbf{v}_s$, the optimal BS power allocation can be obtained via \eqref{optimal_power} by substituting $|c_k(\mathbf{v}_s)|^2$ with $\mathbb{E}\{|c_k(\mathbf{v}_s)|^2\}$.

\textbf{Optimize} $\mathbf{v}_s$ \textbf{for given} $\bar{p}_k$: To tackle the binary constraint (\ref{P_mode}), we first transform it equivalently into the following equality constraint:
\begin{equation}\label{con_binary}
	|[\mathbf{v}_s]_m|(1-|[\mathbf{v}_s]_m|)=0,\forall{s\in\{t,r\},m\in\mathcal{M}}. 
\end{equation}
As $|[\mathbf{v}_s]_m|\leq1$, we can readily observe that $|[\mathbf{v}_s]_m|-|[\mathbf{v}_s]_m|^2\geq0$ and the equality holds only when $|[\mathbf{v}_s]_m|$ is either 0 or 1. Hence, the constraint (\ref{con_binary}) is satisfied only for binary variables. We then add the left-hand-side (LHS) of (\ref{con_binary}) as a penalty term into the objective function, yielding the following optimization problem:
\begin{subequations}
	\begin{align}
		(\text{P2.1}): \mathop  {\max}\limits_{\{\mathbf{v}_s\}}\;\; &\sum_{k=1}^K\bar{R}_k-\eta\sum_{m=1}^M\sum_{s\in\{t,r\}}(|[\mathbf{v}_s]_m|-|[\mathbf{v}_s]_m|^2) \\
		{\rm{s.t.}}\;\;	
		&\eqref{P_energy},\eqref{P1_rate},\eqref{P1_do}, \nonumber
	\end{align}
\end{subequations}
where $\eta>0$ is the penalty factor used to penalize the objective function if $|[\mathbf{v}_s]_m|\in(0,1)$. However, the added penalty term in the objective function is non-convex. To address it, we approximate this term by employing the first-order Taylor series (FTS). Specifically, an upper-bound for the penalty term is obtained as follows:
\vspace{-0.5em}
\begin{equation}
	|[\mathbf{v}_s]_m|-|[\mathbf{v}_s]_m|^2\leq|[\mathbf{v}_s]_m|-\left[2\text{Re}([\widetilde{\mathbf{v}}_s]_m^*[\mathbf{v}_s]_m)-|[\widetilde{\mathbf{v}}_s]_m|^2\right]\triangleq{f_0(\mathbf{v}_s)}.
\end{equation}
By replacing the penalty term with $f_0(\mathbf{v}_s)$, problem (P2.1) can be transformed into
\begin{subequations}
	\begin{align}
		(\text{P2.2}): \mathop  {\max}\limits_{\{\mathbf{v}_s\}}\;\; &\sum_{k=1}^K\bar{R}_k-\eta\sum_{m=1}^M\sum_{s\in\{t,r\}}f_0(\mathbf{v}_s) \\
		{\rm{s.t.}}\;\;	
		&\eqref{P_energy},\eqref{P1_rate},\eqref{P1_do}.\nonumber
	\end{align}
\end{subequations}

Next, to tackle the non-convex objective function and constraint (\ref{P1_rate}), we first introduce auxiliary variables $\{\chi_k\}$ and substitute them into (P2.2). By doing so, the optimization problem can be equivalently expressed as
\vspace{-0.5em}
\begin{subequations}
	\begin{align}
		(\text{P2.3}):\mathop  {\max}\limits_{\{\mathbf{v}_s,\chi_k\}}\;\;  &\sum_{k=1}^K\log_2(1+\chi_k)-\eta\sum_{m=1}^M\sum_{s\in\{t,r\}}f_0(\mathbf{v}_s)  \\
		\label{P2_rate}{\rm{s.t.}}\;\;	
		&\log_2(1+\chi_k) \geq \gamma_k,\forall{k},\\
		\label{P2_auxiliary}&\chi_k\leq\frac{p_k{\mathbb{E}\{c_k(\mathbf{v}_s)\}}}{\sum_{j>k}p_{j}\mathbb{E}\{c_j(\mathbf{v}_s)\}+\sigma^2},\forall{k},\\
		&\eqref{P_energy},\eqref{P1_do}.\nonumber
	\end{align}
\end{subequations}
However, the constraints (\ref{P2_auxiliary}) and (\ref{P1_do}) are still non-convex. To address it, we first rewrite constraint (\ref{P2_auxiliary}) as
\begin{equation}\label{P2_auxiliary_rewrite}
	\frac{p_k(l_k+{\epsilon_k\left|  {\bar{\mathbf{w}}_k^H\mathbf{v}_s}\right|^2+\zeta_k\sum_{m=1}^{M}|[\mathbf{v}_s]_m|^2})}{\chi_k}\geq{\sum_{j>k}p_{j}\mathbb{E}\{c_j(\mathbf{v}_s)\}+\sigma^2}.
\end{equation}
Notice that in the LHS of (\ref{P2_auxiliary_rewrite}), $\frac{\left| {\bar{\mathbf{w}}_k^H\mathbf{v}_s}\right|^2}{\chi_k}$ and $\frac{|[\mathbf{v}_s]_m|^2}{\chi_k}$  are quadratic-over-affine functions, and hence are jointly convex over the involved variables. Besides, $\frac{1}{\chi_k}$ is convex with respect to $\chi_k$.  Thus, we approximate the non-convex terms in the LHS of \eqref{P2_auxiliary_rewrite} by the successive convex approximation (SCA) method. Specifically, according to the FTS, these terms can be approximated at given local points $\widetilde{\mathbf{v}}_s,\widetilde{\chi}_k$ as 
\vspace{-0.5em}
\begin{subequations}
	\begin{align}
		&\chi_k\geq\frac{2}{\widetilde{\chi}_k}-\frac{\chi_k}{(\widetilde{\chi}_k)^2}\triangleq{f_1(\chi_k)},\\
		&\frac{\left|  {\bar{\mathbf{w}}_k^H\mathbf{v}_s}\right|^2}{\chi_k}\geq\frac{2\text{Re}(\widetilde{\mathbf{v}}_s^H{\bar{\mathbf{w}}_k\bar{\mathbf{w}}_k^H}\mathbf{v}_s)}{\widetilde{\chi}_k}-\left(\frac{|\bar{\mathbf{w}}_k^H\widetilde{\mathbf{v}}_s|^2}{\widetilde{\chi}_k}\right)^2\chi_k\triangleq{f_2}(\mathbf{v}_s,\chi_k),\\
		&\sum_{m=1}^{M}\frac{|[\mathbf{v}_s]_m|^2}{\chi_k}\geq\sum_{m=1}^{M}\frac{2\text{Re}([\widetilde{\mathbf{v}}_s]_m^*[\mathbf{v}_s]_m)}{\widetilde{\chi}_k}-\sum_{m=1}^{M}\left(\frac{|[\widetilde{\mathbf{v}}_s]_m|^2}{\widetilde{\chi}_k}\right)^2\chi_k\triangleq{f_3}(\mathbf{v}_s,\chi_k).
	\end{align}
\end{subequations}

Finally, for handling the non-convex constraint (\ref{P1_do}), the LHS can also be replaced by its lower-bound, which is given by
\begin{equation}
	\mathbb{E}\{c_k(\mathbf{v}_s)\}\geq{l_k+\epsilon_k\left[2\text{Re}(\widetilde{\mathbf{v}}_s^H{\bar{\mathbf{w}}_k\bar{\mathbf{w}}_k^H}\mathbf{v}_s)-|\bar{\mathbf{w}}_k^H\widetilde{\mathbf{v}}_s|^2\right]+\zeta_k\sum_{m=1}^{M}\left[2\text{Re}([\widetilde{\mathbf{v}}_s]_m^*[\mathbf{v}_s]_m)-|[\widetilde{\mathbf{v}}_s]_m|^2\right]}\triangleq{f_{4,k}}(\mathbf{v}_s).
\end{equation}

Based on the above approximations, (P2.3) can be transformed into the following problem:
\vspace{-0.2em}
\begin{subequations}\label{p1.4}
	\begin{align}
		(\text{P2.4}):\mathop  {\max}\limits_{\{\mathbf{v}_s,\chi_k\}}\;\;  &\sum_{k=1}^K\log_2(1+\chi_k)-\eta\sum_{m=1}^M\sum_{s\in\{t,r\}}f_0(\mathbf{v}_s)  \\
		\label{dddd}{\rm{s.t.}}\;\;	
		&\log_2(1+\chi_k) \geq \gamma_k,\forall{k},\\
		\label{dddddd}&p_k\left[l_k{f_1(\chi_k)}+\epsilon_k{f_2}(\mathbf{v}_s,\chi_k)+\zeta_k{f_3}(\mathbf{v}_s,\chi_k)\right]\geq{\sum_{j>k}p_{j}\mathbb{E}\{c_j(\mathbf{v}_s)\}+\sigma^2},\forall{k},\\
		&{f_{4,k}}(\mathbf{v}_s)\geq\mathbb{E}\{c_j(\mathbf{v}_s)\},\forall{k<j},\\
		&\eqref{P_energy}.\nonumber
	\end{align}
\end{subequations}

It is easy to verify that problem (P2.4) is a convex optimization problem, which thus can be efficiently solved via standard solvers such as CVX\cite{cvx}.

\emph{2) Outer Layer Iteration:}
In the outer layer, we gradually increase the value of $\eta$ as $\eta=\omega\eta$ with $\omega>1$. The iteration of the outer loop terminates when the violation of the equality constraint is within a predefined accuracy, i.e., 
\begin{equation}
	\max\left\lbrace |[\mathbf{v}_s]_m|-|[\mathbf{v}_s]_m|^2,\forall{m\in\mathcal{M},s\in\{t,r\}}\right\rbrace \leq\epsilon.
\end{equation}

\begin{algorithm}[!t]\label{method1}
	\caption{Proposed penalty-based iterative algorithm for solving problem (P2)}
	\begin{algorithmic}[1]
		\STATE {Initialize feasible points $\left\{ \mathbf{v}_s^{(n)}, \chi_k^{(n)} \right\}$ with $n=0$, the penalty factors $\eta$.}
		\STATE {\bf repeat: outer loop}
		\STATE \quad {\bf repeat: inner loop}
		\STATE \quad\quad For given $\left\{  \mathbf{v}_s^{(n)}, \chi_k^{(n)}\right\}$, obtain the power allocation $\bar{p}_k^{(n+1)}$ via \eqref{optimal_power}.
		\STATE \quad\quad For given $\bar{p}_k^{(n+1)}$, update $\left\{  \mathbf{v}_s^{(n
			+1)}, \chi_k^{(n+1)}\right\}$ by solving \eqref{p1.4}.
		\STATE \quad\quad  $n=n+1$.
		\STATE \quad {\bf until} the maximum number of inner iterations $n_{\max}$ is reached.
		\STATE \quad Update the penalty factor as $\eta  = \omega \eta $.
		\STATE {\bf until} the constraint violation is below a predefined threshold $\epsilon>0$.
	\end{algorithmic}
\end{algorithm}

The overall algorithm for solving (P2) is summarized in \textbf{Algorithm 1}. Since the value of the objective function is bounded and non-decreasing in each iteration of the inner loop, the proposed penalty-based iterative algorithm is guaranteed to convergence as $\eta$ grows sufficiently large. Next, we analyze its computational complexity. Specifically, let $I_{\text{out}}$ and $I_{\text{in}}$ denote the number of outer and inner iterations. As solving the inner problem requires the computation complexity order of $\mathcal{O}((KM)^{3.5})$, the whole computational complexity order of \textbf{Algorithm 1} is given by $\mathcal{O}(I_{\text{out}}I_{\text{in}}(KM)^{3.5})$\cite{convex}.

\vspace{-0.5em}
\begin{remark}
\rm It is worth mentioning that the performance of the proposed BTE protocol is affected by the Rician factors, as explained below. From (\ref{expect_channel}), it is observed that the expected effective channel power gains (i.e., $\mathbb{E}\{|c_k(\mathbf{v}_s)|^2\}$) are determined by $\{{\mathbf{v}}_s\}$ and the Rician factors $\kappa_{1},\kappa_{2}$.
Specifically, when $\kappa_{1},\kappa_{2}$ are large (corresponding to the case of LoS-dominant channels), $\mathbb{E}\{|c_k(\mathbf{v}_s)|^2\}$ is proportional to $\left|\bar{\mathbf{w}}_k^H\mathbf{v}_s\right|^2$.
On the other hand, when $\kappa_{1},\kappa_{2}\rightarrow0$ (corresponding to Rayleigh fading), the second term in (\ref{expect_channel}) (i.e., $\epsilon_k\left|  {\bar{\mathbf{w}}_k^H\mathbf{v}_s}\right|^2$) will have little impact on the channel power gains since $\epsilon_k\rightarrow0$. In this case, the effective channel power gains are mainly determined by the direct channels and the number of transmission/reflection elements, whereas the STAR-RIS phase-shifts do not affect the system performance. In other words, the STAR-RIS serves as a random scatterer, whose phase-shifts can be set arbitrarily. Therefore, for the rich scattering environment, the BTE protocol may be ineffective since the system performance gain brought by STAR-RIS is limited by the dominant NLoS channel components. This motivates us to devise an alternative scheme in the next section to fully exploit the aperture gain of STAR-RIS when the Rician factors are small.	
\end{remark}

\vspace{-1em}
\section{TTS Transmission Protocol for Rich Scattering Environments}
In this section, we consider a rich scattering environment with small Rician fading factors. For this scenario, we propose a customized TTS transmission protocol and devise an efficient algorithm to design the STAR-RIS coefficients and BS power allocation.
\vspace{-1em}
\subsection{Proposed PTE Transmission Protocol}

One key observation is that if each STAR-RIS element is allocated for serving one user, the instantaneous channel estimation overhead will be greatly reduced. Leveraging this observation, we propose a new PTE transmission protocol as follows. In the first phase, the BS acquires the large-scale path-loss information (i.e., $\delta_{BS},\delta_{Sk},\delta_{k}$). Based on this S-CSI, the BS determines the STAR-RIS surface-partition strategy where the STAR-RIS elements are divided into several transmission/reflection 
subsurfaces, each of which is responsible to assist the data transmission of one specific T/R user. As such, the surface-partition strategy is settled and remain unchanged for the subsequent time slots. In the second phase, for each channel coherence block, the BS estimates \textit{the direct and cascaded subsurface channels} (referred to as the partial I-CSI) for each user instead of the full I-CSI and then optimizes its power allocation as well as the STAR-RIS phase-shifts. In the third phase, the BS proceeds with data transmission.

\begin{figure}[t!]
	\centering
	\includegraphics[width=1\textwidth]{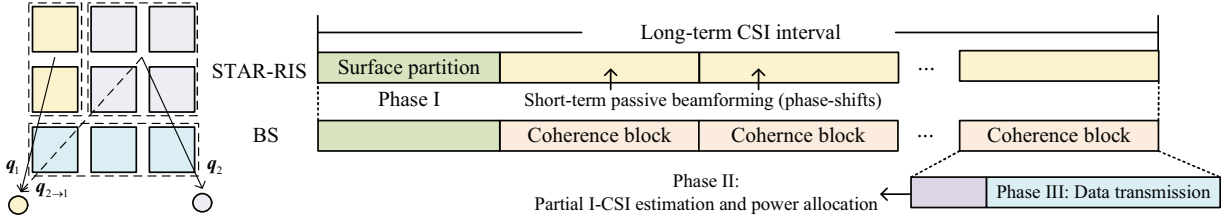}\\
	\caption{Illustration of the proposed PTE protocol.}\vspace{1em}
\end{figure}

\vspace{-1em}
\subsection{Problem Formulation}

As introduced before, we target to partition the STAR-RIS elements into $K$ subsurfaces for sum-rate maximization. Assume that the $k$-th subsurface containing $M_k$ adjacent elements is allocated to user $k$. Accordingly, we have $\sum_{k=1}^{K}M_k=M$ and the indexes of the STAR-RIS elements associated with user $k$ are denoted by $\mathcal{M}_k=\{1,...,M_k\}$. Let $\mathbf{q}_k^H\in\mathbb{C}^{1\times{M_k}}$ denote the cascaded BS to user $k$ channel through the $k$-th subsurface. For each channel coherence block, it takes $M$ pilot symbols for estimating all the \textit{cascaded subsurface channels} $\{\mathbf{q}_k\}$ and $K$ symbols for estimating all the direct channels $\{h_k\}$. Hence, the instantaneous  channel estimation overhead for the PTE protocol is $M+K$, which is much smaller than for the full I-CSI based protocol (i.e., $KM+K$). Further, we assume that these partial I-CSI can be acquired via existing channel estimation methods (see e.g., \cite{ce_survey,discrete_jsac}), and is perfectly known at the BS. Considering the channel estimation overhead, the average achievable rate of user $k$ is expressed as $R_k^b=(1-\frac{M+K}{T_c})R_k$.

For the PTE protocol, we aim to maximize the average sum-rate by optimizing the long-term STAR-RIS surface-partition strategy based on large-scale path-loss information, and the short-term STAR-RIS phase-shifts and BS power allocation based on the estimated partial I-CSI. This problem can be formulated as follows:
\begin{subequations}
	\begin{align}
		(\text{P3}):\; { \mathop{\max}\limits_{\{M_k\}}}\;&\mathbb{E}\left\lbrace  {\mathop{\max}\limits_{\{p_k,\boldsymbol{\theta}_s\}}}\; \sum_{k=1}^K{R_k^b} \right\rbrace  \\
		{\rm{s.t.}}\;\;	
		&{R}_k^{b} \geq \gamma_k,\forall{k\in\mathcal{K}},\\
		&|c_k(\mathbf{v}_s)|^2>|c_j(\mathbf{v}_s)|^2, \forall{k<j},  \\
		&|\theta_m^s|=1,\forall{m,s},\\
		&\sum_{k=1}^K{p_k}\leq{P_{\text{max}}}.
	\end{align}
\end{subequations}

\subsection{Proposed Solution to Problem (P3)}

\emph{i) Short-term BS power allocation and STAR-RIS phase-shifts design:}

For the short-term optimization problem, we first propose a suboptimal yet low-complexity phase-shift design that is to align the $k$-th cascaded subsurface channel with the direct link of user $k$.\footnote{This phase-shift design may cause a loss in terms of STAR-RIS beamforming gain but is practically efficient since it eliminates the need to estimate the full I-CSI of the whole surface.} 
Specifically, the $n$-th phase-shift of subsurface $k$ is set as 
\begin{equation}\label{opt_phase}
	\theta_n^k=\angle{h_{k}}-\angle[\mathbf{q}_{k}]_n,\forall{n}\in\mathcal{M}_k,
\end{equation}

For NOMA transmission, the received superimposed signal at user $k$ is given by
\begin{equation}
	y_{k}=c_k \sum_{k=1}^{K}\sqrt{p_k}s_k+n_{k},
\end{equation}
with
\begin{equation}
	c_k=|h_{k}|+\sum_{n=1}^{M_k}|[\mathbf{q}_{k}]_n|+\sum_{j\atop{j\neq{k},j,k\in\mathcal{K}_s}}\sum_{n=1}^{M_j}[\mathbf{q}_{j\rightarrow{k}}]_n\overset{(a)}\approx |h_{k}|+\sum_{n=1}^{M_k}|[\mathbf{q}_{k}]_n|,
\end{equation}
where $\mathbf{q}^H_{j\rightarrow{k}}\in\mathbb{C}^{1\times{M_j}}$ denotes the cascaded channel between the $j$-th subsurface and user $k$.
The approximation in $(a)$ is follows by neglecting the signals transmitted/reflected by other subsurfaces that are not assigned to user $k$ (i.e., the third term in $c_k$). In practical scenarios, the T/R users are exposed to all transmission/reflection elements. For example, for user $k\in\mathcal{K}_t$, the subsurfaces allocated to other users $j\in\mathcal{K}_t$ ($j\neq{k}$) will also scatter signals, which can be combined either constructively or destructively with user $k$'s desired signals. Fortunately, our approximation is practically accurate since the effect of the scattered signals is insignificant, especially when $\kappa_{1}$ and $\kappa_{2}$ are small, which will be numerically verified in Section V.

Next, regarding $c_k$ as the effective channel of each user, the BS determines the short-term power allocation based on \textbf{Proposition 1} after estimating the partial I-CSI.

\emph{ii) Long-term surface-partition optimization:}

Next, we optimize the long-term surface-partition for the average sum-rate maximization. Applying the phase-shift design in \eqref{opt_phase}, the expected effective channel power gains $\mathbb{E}\{(c_k)^2\}$ are characterized as follows.
\vspace{-0.5em}
\begin{lemma}
For the PTE protocol, the expected effective channel power gain of user $k$ is given by
	\begin{equation}\label{expected_channel_gain_scheme2}
		\begin{aligned}
			\mathbb{E}\{(c_k)^2\}=&
			\delta_k+\delta_{BS}\delta_{Sk}{M_k}+\frac{\pi^2\delta_{BS}\delta_{Sk}M_k(M_k-1)}{16(\kappa_{1}+1)(\kappa_{2}+1)}\left(L_{\frac{1}{2}}(-\kappa_{1})L_{\frac{1}{2}}(-\kappa_{2})\right)^2+\\
			&\sqrt{\frac{\pi^3\delta_{Bk}\delta_{BS}\delta_{Sk}}{16(\kappa_{1}+1)(\kappa_{2}+1)}}L_{\frac{1}{2}}(-\kappa_{1})L_{\frac{1}{2}}(-\kappa_{2}){M_k},
		\end{aligned}	
	\end{equation}
	where $L_{\frac{1}{2}}(\cdot)$ denotes the Laguerre polynomial.
\end{lemma}
\vspace{-1em}
\begin{proof}
	See Appendix C.
\end{proof}


Based on \textbf{Lemma 3}, the expected achievable rate of user $k$ can be approximated as
\begin{equation}\label{rate_scheme2}
	\mathbb{E}(R_k^b)\approx\left(1-\frac{M+K}{T_c}\right)\log_2\left( 1+\frac{p_k{\mathbb{E}\{(c_k)^2\}}}{\sum_{j>k}p_{j}\mathbb{E}\{(c_j)^2\}+\sigma^2}\right)\triangleq {\breve{R}_k}.
\end{equation}

Therefore, the long-term surface-partition optimization problem can be approximated as 
\begin{subequations}
	\begin{align}
		(\text{P4}): &\mathop  {\max}\limits_{\{M_k\}}\;\; \sum_{k=1}^{K} \breve{R}_k \\
		{\rm{s.t.}}\;\;	
		& \breve{R}_k \geq \gamma_k,\\
		\label{P2_do}&\mathbb{E}\{(c_k)^2\}>\mathbb{E}\{(c_j)^2\}, \forall{k<j}.
	\end{align}
\end{subequations}
Problem (P4) is a non-convex mixed-integer non-linear programming (MINLP), which is generally hard to solve. Although the optimal solution can be obtained by enumerating all possible surface-partition strategies, this exhaustive-search based algorithm entails a prohibitively high computational complexity, which is unaffordable when $M$ is large. Moreover, the conventional integer relaxation and rounding approach can not be applied, since the relaxed problem is still not convex. To tackle these issues, we propose an efficient two-step heuristic algorithm to suboptimally solve problem (P4), whose details are introduced as follows.

\emph{1) Initial STAR-RIS surface-partition under rate and decoding-order constraints:} It is expected that NOMA transmission benefits more from the multiplexing gain when the channel conditions of users are more distinctive. Hence, in order to improve the NOMA gain, it is desirable to allocate more elements to the users with higher decoding orders to create asymmetric channel conditions. However, with this strategy, the weak users' rate requirements may not be met if their effective channel power gains are too small. To tackle this, in the first step, we aim to determine the minimum number of elements allocated to each user for satisfying their rate requirements and the decoding order constraints. Specifically, we first select a small number of elements (e.g., the minimum number of elements in a group) to be assigned to user $K$ (the weakest user), denoted by $M_K^{(0)}$. Based on $M_K^{(0)}$, we obtain the minimum number of elements assigned to other users, i.e.,  $\{M_{k}^{(0)}\}$ that satisfy $\mathbb{E}\{|c_k|^2\}>\mathbb{E}\{|c_{k+1}|^2\},\forall k<K$; while the remaining elements are allocated to user 1 which is of the highest decoding order. Next, we decide whether the rate requirements are satisfied under the current surface-partition strategy by checking the feasibility constraint in \eqref{optimal_power}. If it is infeasible, we gradually increase $M_k^{(0)}$ by $M_k^{(0)}=M_k^{(0)}+1$ and repeat the above steps until all users' target rates are reached.

\begin{algorithm}[!t]\label{method2}
	\caption{Proposed low-complexity algorithm for solving problem (P4)}
	\begin{algorithmic}[1]
		\STATE {\bf Step 1:}
		\STATE {Initialize $t=0$, a small number of elements assigned to user $K$, $M_K^{(t)}$.}
		\STATE Calculate the minimal $M_k^{(t)}$ that satisfies the decoding order constraint.
		\STATE \quad {\bf repeat}   
		\STATE \quad\quad  $M_k^{(t+1)}=M_k^{(t)}+1,\forall{k}.$
		\STATE \quad\quad $t=t+1$.
		\STATE \quad {\bf until} \eqref{optimal_power} is feasible.
		\STATE {\bf Step 2:}
		\STATE {Initialize $l=0$, $M_k^{(l)}$ as that obtained from the last step.}
		\STATE \quad {\bf repeat}   
		\STATE \quad\quad $M_1^{(l')}=M_1^{(l)}-1$, $M_k^{(l')}=M_k^{(l')}+1,k=2,...,K$.
		\STATE \quad\quad {\bf if} $R^{(l')}>R^{(l)}$ {\bf then}
		\STATE \quad\quad \quad $l=l+1.$
		\STATE \quad\quad \quad $M_k^{(l)}=M_k^{(l')},p_k^{(l)}=p_k^{(l')},\forall{k}.$
		\STATE \quad\quad {\bf endif} 
		\STATE \quad {\bf until} there is no rate increase.
	\end{algorithmic}
\end{algorithm}

\emph{2) Progressive surface-partition refinement for sum-rate enhancement:} Although the rate requirements are satisfied in the first step, the obtained surface-partition strategy may not be the optimal choice for sum-rate maximization. To improve the solution, we give the following lemma.
\vspace{-0.5em}
\begin{lemma}
\rm Given the initial surface-partition strategy in Step 1, and BS power allocation in \eqref{optimal_power}, if an element assigned to user $k$ is reassigned to another user $k'$, then $R_{k'}^{b}$ increases and $R_k^b$ decreases, while the achievable rates of the other users remain unchanged. 
\end{lemma}
\vspace{-0.5em}
\begin{proof}
	First, it is observed from the expression of $R_k^b$ that the expected channel power gain and the transmit power of other users are not affected by the reassignment. Therefore, their average achievable rates do not change. Second, after this reassignment, the expected channel power gain of user $k$ and $k'$, i.e., $\mathbb{E}\{(c_k)^2\}$ and $\mathbb{E}\{(c_j)^2\}$, will increase and decrease, respectively. It is observed from \eqref{rate_scheme2} that the rate of user $k$ monotonically increases with $\mathbb{E}\{(c_k)^2\}$. Hence, the average rates of users $k$ and $k'$ follow the same trend as their expected channel power gains. This completes the proof.
\end{proof}

\textbf{Lemma 4} implies that the user rate can be adjusted by reallocating the elements from one to another user. This inspires us to design a progressive surface-partition refinement algorithm to gradually improve the system sum-rate. Specifically, we denote the initial surface-partition strategy obtained in Step 1 as $\{M_k^{(l)}\}$, where $l$ is the initial iteration number; and obtain the power allocation and the corresponding average sum-rate, denoted by $\{p_k^{(l)}\}$ and  $R^{(l)}$, respectively. Next, we consecutively reassign an element associated with user 1 to another user $k, k=2,3...,K$, and calculate the average sum-rate after reassignment, denoted by $R^{(l')}$. If $R^{(l')}>R^{(l)}$, we update the STAR-RIS surface-partition strategy and BS power allocation, and then enter the next iteration. This process is repeated until the system sum-rate can not be improved anymore.

The main procedures of the proposed solution for solving (P4) are summarized in \textbf{Algorithm 2}. Generally, the worst-case complexity of the proposed algorithm is $\mathcal{O}(KM)$, which is much small than that of the exhaustive-search based algorithm, i.e. $\mathcal{O}(M^K)$.

\begin{table}[]
	\caption{Comparison between the proposed two TTS transmission protocols.}
	\label{tab:my-table}
	\resizebox{\textwidth}{!}{%
		\begin{tabular}{|c|c|c|c|c|c|}
			\hline
			\textbf{Protocol}    & \textbf{CSI interval} & \textbf{Required CSI}                                                               & \textbf{Design parameters}                                                           & \textbf{Method}                                                    & \textbf{Overhead} \\ \hline
			\multirow{2}{*}{BTE} & long-term             & \begin{tabular}[c]{@{}c@{}}path-loss, Rician factor, \\ LoS components\end{tabular} & STAR-RIS coefficients                                                                & Algorithm 1                                                        & \textbackslash{}  \\ \cline{2-6} 
			& short-term            & effective channels                                                                  & BS power allocation                                                                  & Proposition 1                                                      & $K$                 \\ \hline
			\multirow{2}{*}{PTE} & long-term             & path-loss                                                                           & surface partition strategy                                                           & Algorithm 2                                                        & \textbackslash{}  \\ \cline{2-6} 
			& short-term            & subsurface cascaded channels                                                        & \begin{tabular}[c]{@{}c@{}}STAR-RIS phase-shifts \\ BS power allocation\end{tabular} & \begin{tabular}[c]{@{}c@{}}Eq. (21), \\ Proposition 1\end{tabular} & $M+K$               \\ \hline
		\end{tabular}%
	}
\end{table}

\vspace{-1em}
\subsection{Comparison between the proposed protocols}
In Table I, we summarize and compare the proposed two TTS transmission protocols in terms of the required CSI, design parameters, optimization methods, and channel estimation overhead. Compared with the BTE protocol, the PTE protocol requires a larger channel estimation overhead, while is less sensitive to the variations of channel statistics.
Specifically, the BTE protocol is only effective when the channels are LoS-dominant, while results in a considerable rate loss in rich scattering environments, e.g., when $\kappa_{1},\kappa_{2}\rightarrow0$, as explained in \textbf{Remark 1}. In contrast, the proposed PTE protocol provides a squared-order passive beamforming gain, i.e., $\mathcal{O}(M^2)$, even for Rayleigh fading channels. This is because the expected effective channel power gain, i.e., $\mathbb{E}\{(c_k)^2\}$ always contains the term (i.e., the third term in \eqref{expected_channel_gain_scheme2}) that is proportional to $(M_k)^2$, as shown in \textbf{Lemma 3}.

\section{Simulation Results}
In this section, numerical results are provided to validate the effectiveness of our proposed TTS transmission protocols and the STAR-RIS MS design. We consider a three-dimensional system setup, where the BS is located at the (0, 0, 1) meters (m) and the reference center of the STAR-RIS is at (50, 0, 1) m. 
We assume that the T and R users are randomly located within a circle centered at the STAR-RIS with a radius of $5$ m. The distance-dependent path loss is modeled as $\delta(d)=\rho_0(\frac{d}{d_0})^{-\alpha}$, where $d$ is the individual link distance, $\rho_0$ denotes the path loss at the reference distance $d_0=1$ m, and $\alpha$ denotes the path loss exponent. We set $\rho_0=-30$ dB, $\alpha_{k}=3.5$, $\alpha_{BS}=2$, $\alpha_{Sk}=2.2$. The noise power is set as $\sigma^2 = -80$ dBm and the Rician factors are set as $\kappa_{1}=\kappa_{2}=\kappa$. Besides, the achieved sum-rates of the proposed protocols are averaged over 5000 random channel realizations.

\begin{figure}[t!]
	\centering
	\subfigure[Convergence behavior of Algorithm 1 for the BTE protocol with $M=40$ and $M=80$.]{\label{con1}
		\includegraphics[width= 3in]{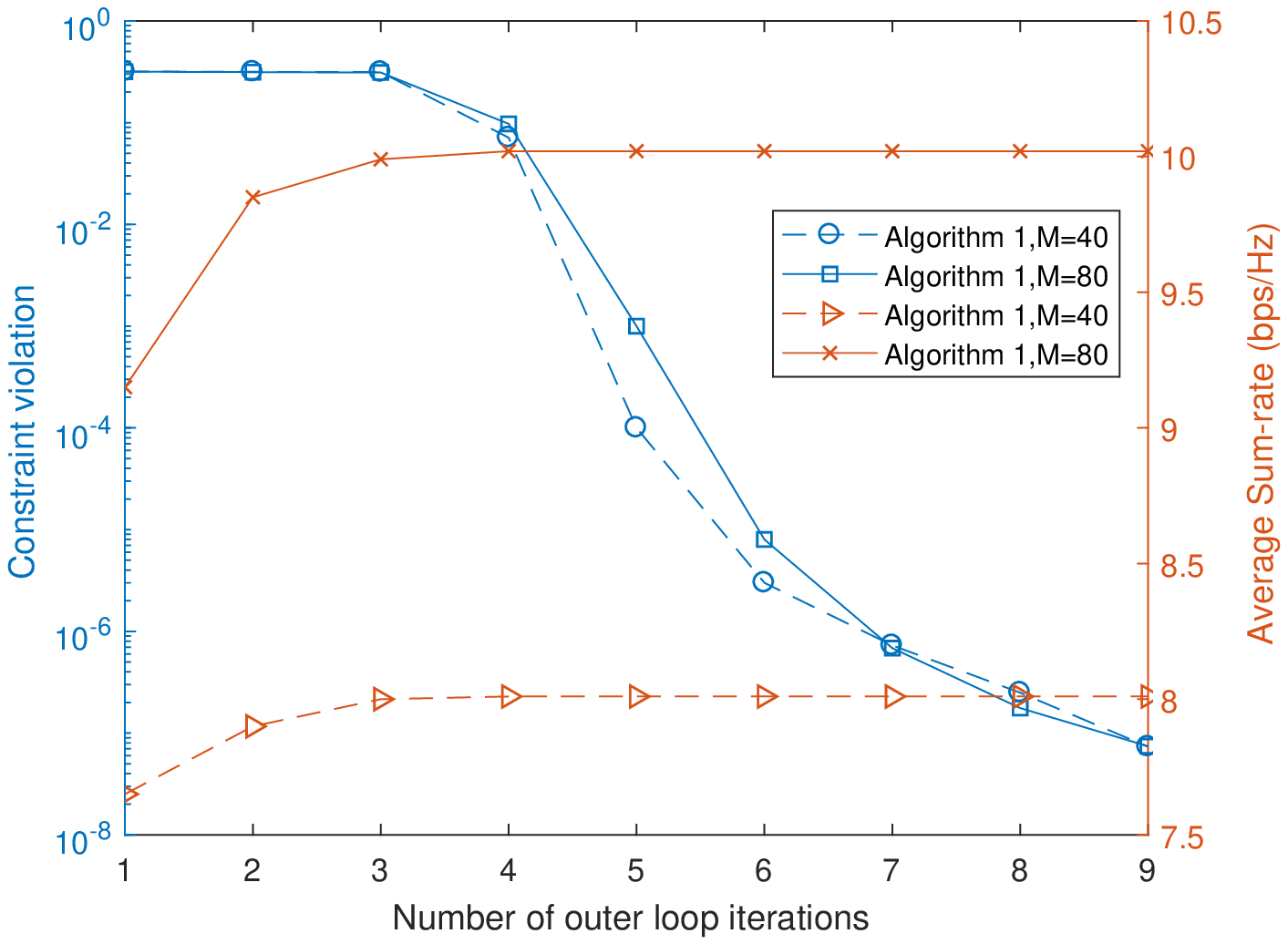}}
	\subfigure[Convergence behavior of Algorithm 2 for the PTE protocol with $M=50$ and $M=60$.]{\label{con2}
		\includegraphics[width= 3in]{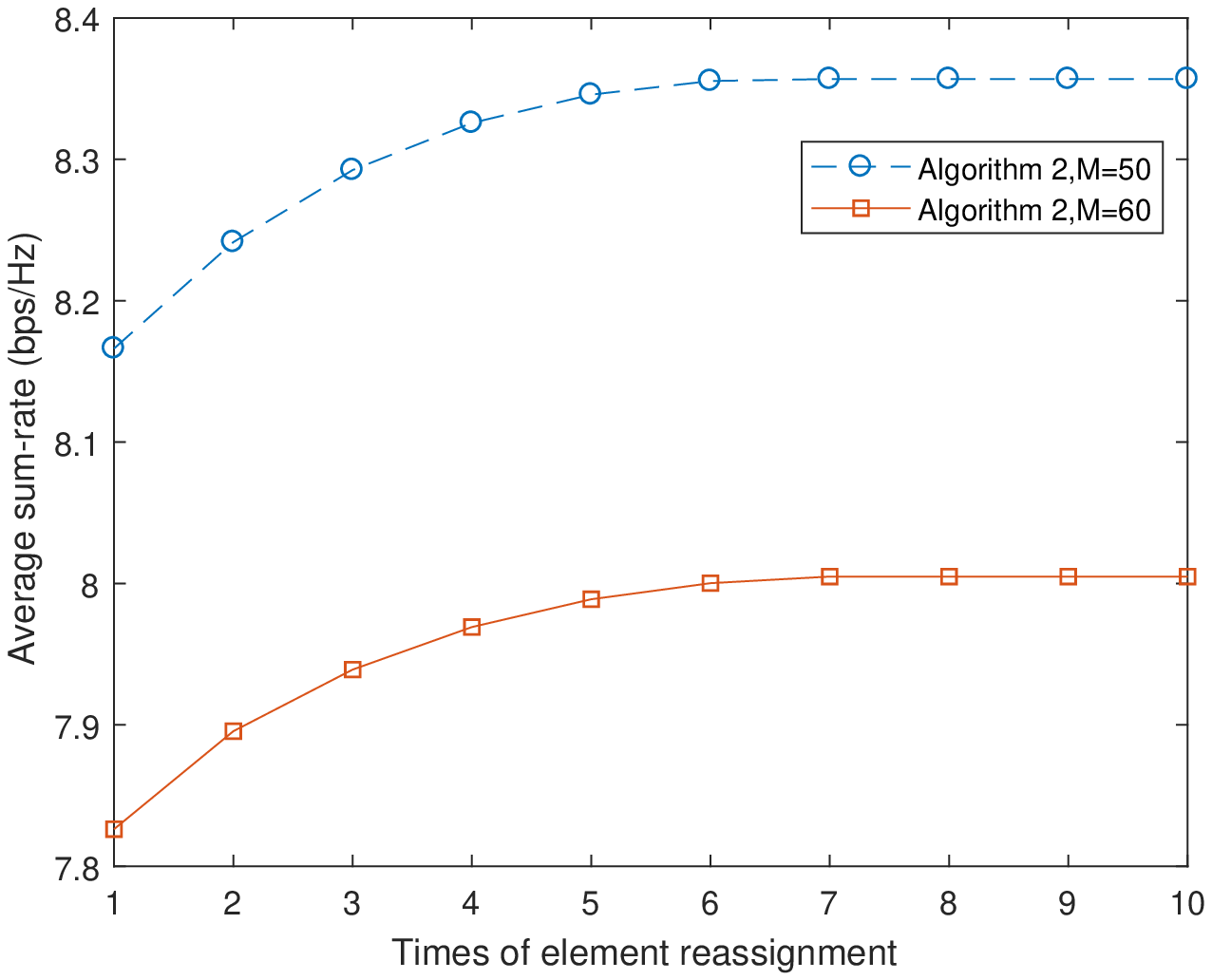}}
	\caption{Convergence behavior of proposed algorithms with $T_c=400,\kappa=1, \gamma_k=2$ bps/Hz.}\label{con}
\end{figure}

\emph{1) Convergence behavior of the proposed algorithms:} In Fig. \ref{con}, we verify the convergence behavior of the proposed algorithms. Specifically,  Fig. \ref{con1} shows the violation of the equality constraints and the average sum-rate of \textbf{Algorithm 1} versus the number of outer loop iterations with $\eta=10^{-4}$ and $\omega=10$. It is observed that the constraint violation decreases quickly as the number of outer loop iterations increases, and ultimately reaches the predefined accuracy (i.e., $\epsilon=10^{-7}$). Besides, the sum-rate increases quickly and converges after 4 iterations. Fig. \ref{con2} plots the system sum-rate versus the times of element reassignments for the second step of \textbf{Algorithm 2}. With the proposed progressive refinement algorithm, the system sum-rate is observed to monotonically increase once the rate requirements are met in the first step. 

\emph{2) Accuracy of the expected rate/channel power gain approximation:} In Fig. \ref{fig_rate_robust}, we evaluate the accuracy of the approximation for the expected sum-rate of the proposed protocols, i.e., $\sum_{k=1}^K\bar{R}_k$ and $\sum_{k=1}^K\breve{R}_k$. Specifically, the approximations of the expected rates of each user for the proposed BTE protocol and PTE protocol are given as \eqref{noma_expect_rate_approx} and \eqref{rate_scheme2}, respectively. Two sets of experiments are carried out under $P_{\max}=30$ dBm, $\gamma_k=2$ bps/Hz and $P_{\max}=25$ dBm, $\gamma_k=1$ bps/Hz, respectively. It is observed that, for all cases and both protocols, the approximated sum-rate is only slightly larger than the simulation result, since the approximation using the Jensen equality serves as an upper-bound for the exact average achievable rate.

In Fig. \ref{fig_scheme2_robust}, we further evaluate the accuracy of the approximation for the expected channel power gains of the PTE protocol given in \eqref{expected_channel_gain_scheme2}. The surface-partition strategy is optimized under $M=50,\gamma_k=2, \kappa=1$. Under this setup, according to \textbf{Algorithm 2}, the number of elements assigned to each user is 36, 9, 3 and 2, respectively.
The exact expected channel power gains, i.e., $|c_k|^2$, are obtained by averaging over 10000 random channel realizations for each user. As observed, our analysis matches well with the exact results for users 1, 3, 4, regardless of the Rician factor $\kappa$. This is because these users' desired signals are almost not affected by other transmission/reflection elements that are not assigned to them. For user 2, its exact channel power gain is slightly larger than the made approximation. This is because, in this case, the subsurface associated with user 1 serves as a random scatter, whose scattered signals are combined constructively with user 2's desired signals.

\begin{figure}[t!]
	\centering
	\begin{minipage}[t]{0.45\linewidth}
		\includegraphics[width=3in]{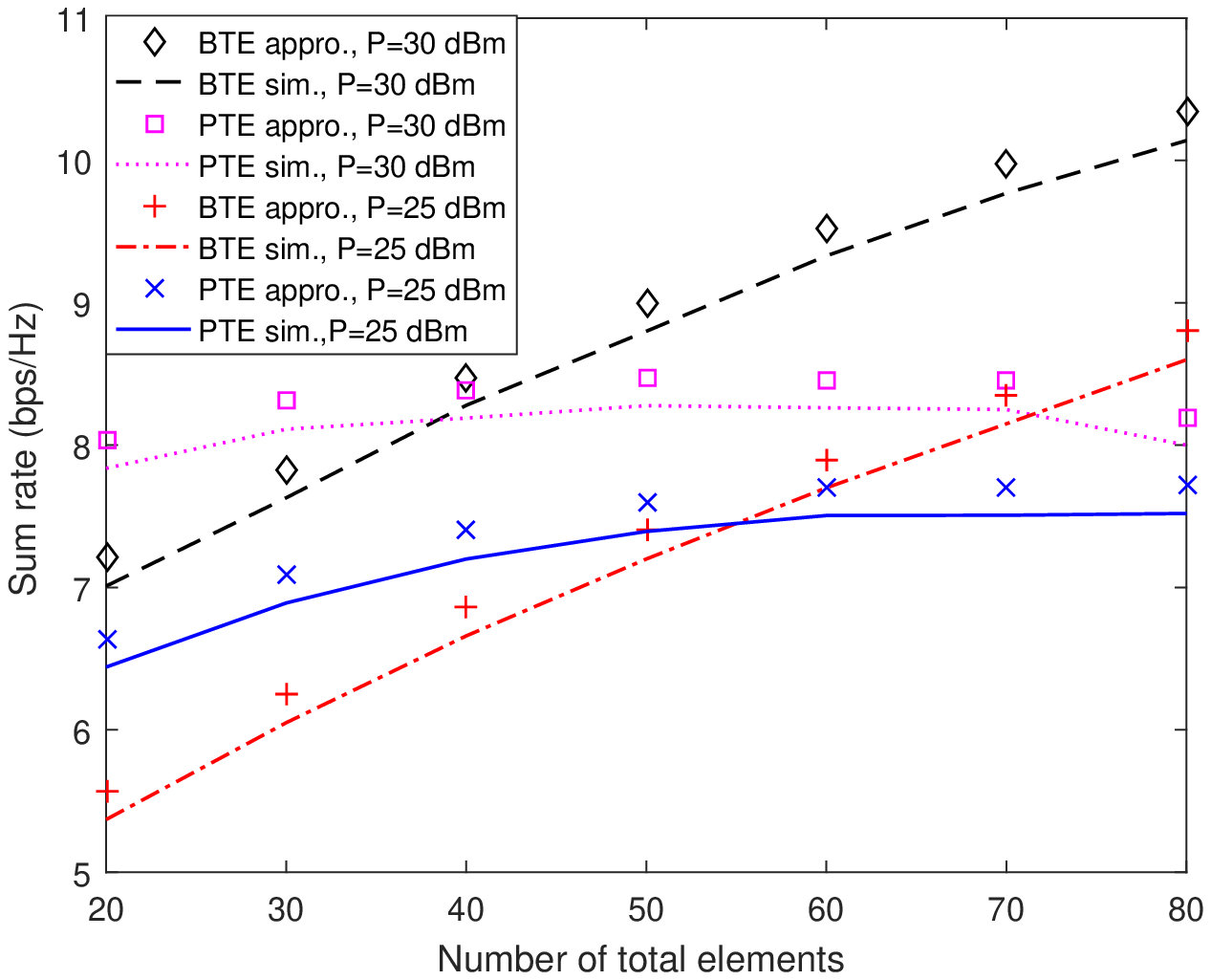}
		\caption{Accuracy of sum-rate approximation.}
		\label{fig_rate_robust}
	\end{minipage}
	\quad
	\begin{minipage}[t]{0.45\linewidth}
		\includegraphics[width=3in]{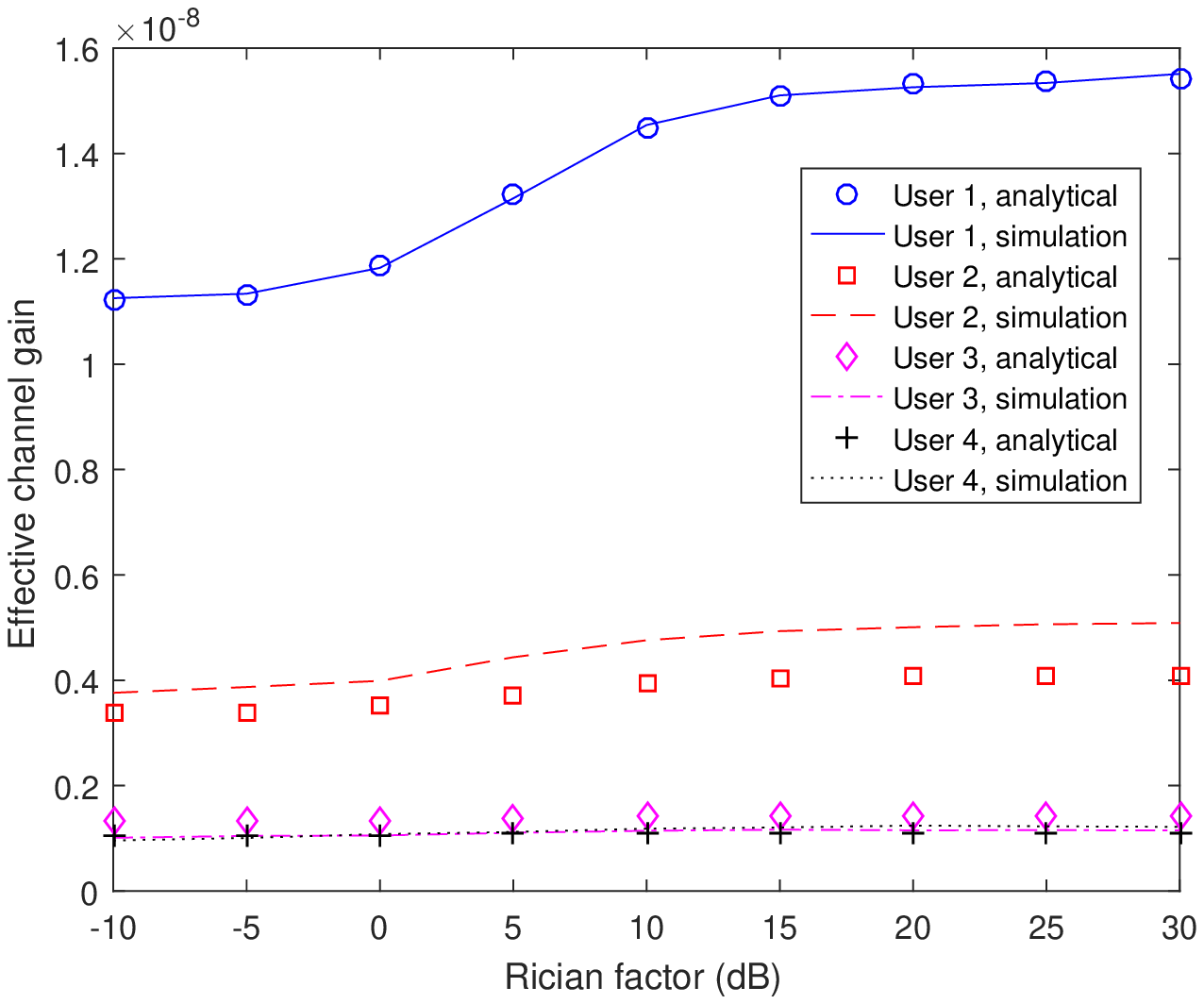}
		\caption{Accuracy of effective channel power gain approximation.}
		\label{fig_scheme2_robust}
	\end{minipage}
\end{figure}

\emph{3) Effectiveness of the proposed algorithms:} To show the effectiveness of our proposed algorithms, we consider the following benchmarks: 1) For the BTE protocol, the semidefinite relaxation (SDR) method is used for optimizing the phase-shifts\cite{irs_joint} and the penalty-based method is adopted for satisfying the binary amplitude constraint. This benchmark serves as a performance upper-bound since the rank-one constraint for phase-shifts is relaxed using SDR. 2) For the PTE protocol, we choose the exhaustive-search based algorithm for surface-partition as a performance upper-bound.  We show the achievable sum-rates attained by different algorithms versus the number of total elements in Fig. \ref{effective}. It is observed that the proposed algorithms for both protocols achieve very close performance with their respective upper-bounds, which thus validates the effectiveness of our proposed algorithms.

\emph{4) Performance comparison:} To demonstrate the benefits brought by STAR-RIS and NOMA, we consider the following benchmark schemes:

\begin{itemize}
	\item \textbf{STAR-RIS assisted FDMA/TDMA system:} Two types of OMA schemes are considered, namely, frequency division multiple access (FDMA) and time division multiple access (TDMA). For FDMA, the operation modes and the phase-shifts are configured in the same way as our proposed PTE protocol, while the users are served over orthogonal frequency RBs. For TDMA, each coherence block is further divided into several equal sub-blocks, each of which is responsible for serving one user, with all the STAR elements operating in the transmission/reflection mode. Note that the channel estimation overhead for FDMA and TDMA is $M+K$ and $K(M+1)$, respectively.
	\item \textbf{Conventional RIS-assisted NOMA system (CR-NOMA):} In this case, one reflecting-only RIS and one transmitting-only RIS (each with $M/2$ elements) are co-located for achieving full-space coverage. This scheme can also be regarded as a special case of mode switching, where half of the elements are fixed to operate in the transmission/reflection mode.
	\item \textbf{Instantaneous full CSI without considering channel estimation overhead (Ideal CSI):} In this case, we assume that the instantaneous full CSI is perfectly known at the BS for system optimization. To evaluate the performance of the proposed CSI acquisition methods, we neglect the channel estimation overhead of this benchmark scheme (which requires $KM+K$ channel training symbols in practice), thus its achievable sum-rate serves as the performance upper-bound.
\end{itemize}

\begin{figure}[t!]
	\centering
	\begin{minipage}[t]{0.45\linewidth}
		\includegraphics[width=3in]{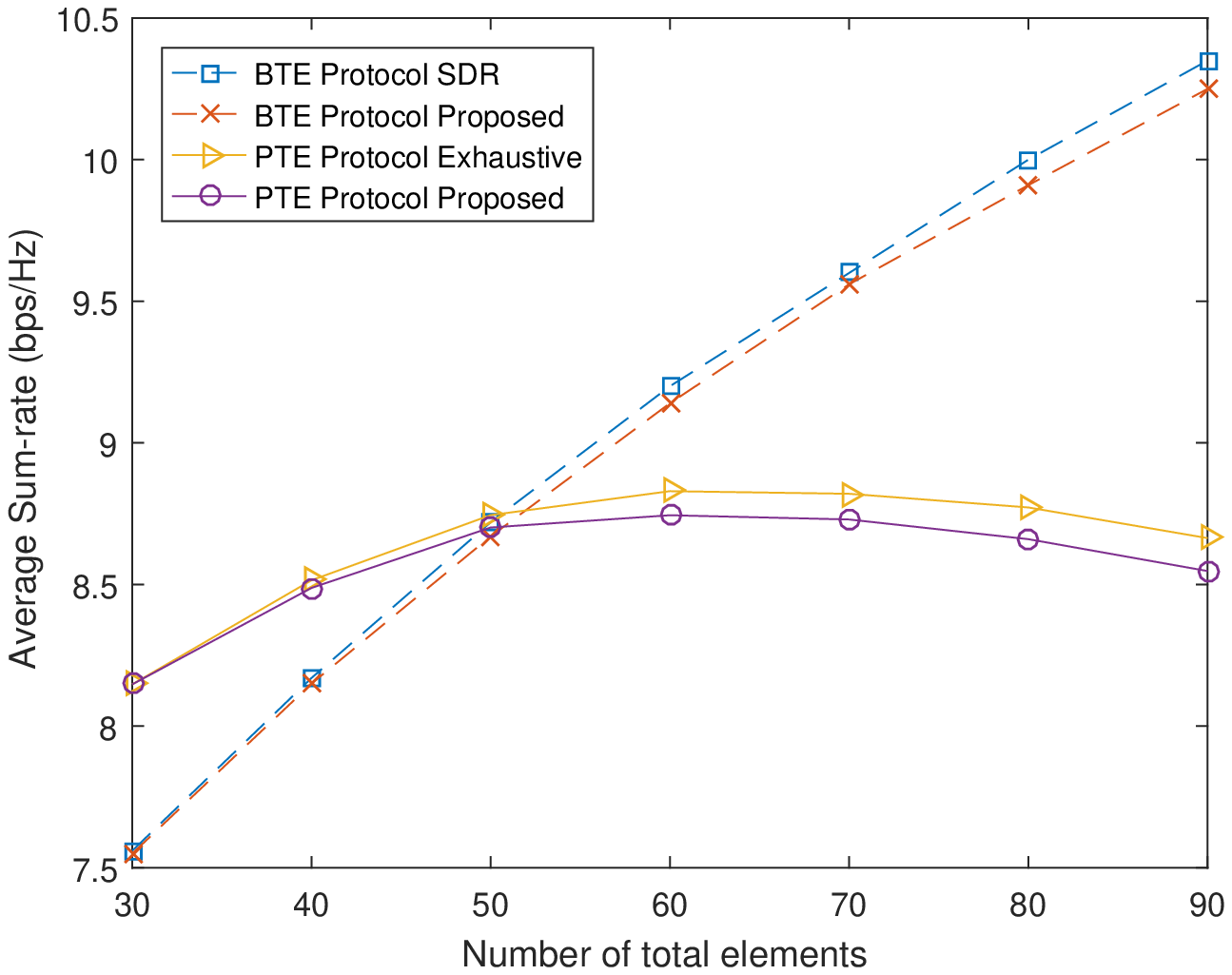}
		\caption{Average achievable sum-rate versus the number of total elements $M$ with $T_c=400, \kappa=3, \gamma_k=1.8$ bps/Hz.}
		\label{effective}
	\end{minipage}
	\quad
	\begin{minipage}[t]{0.45\linewidth}
		\includegraphics[width=3in]{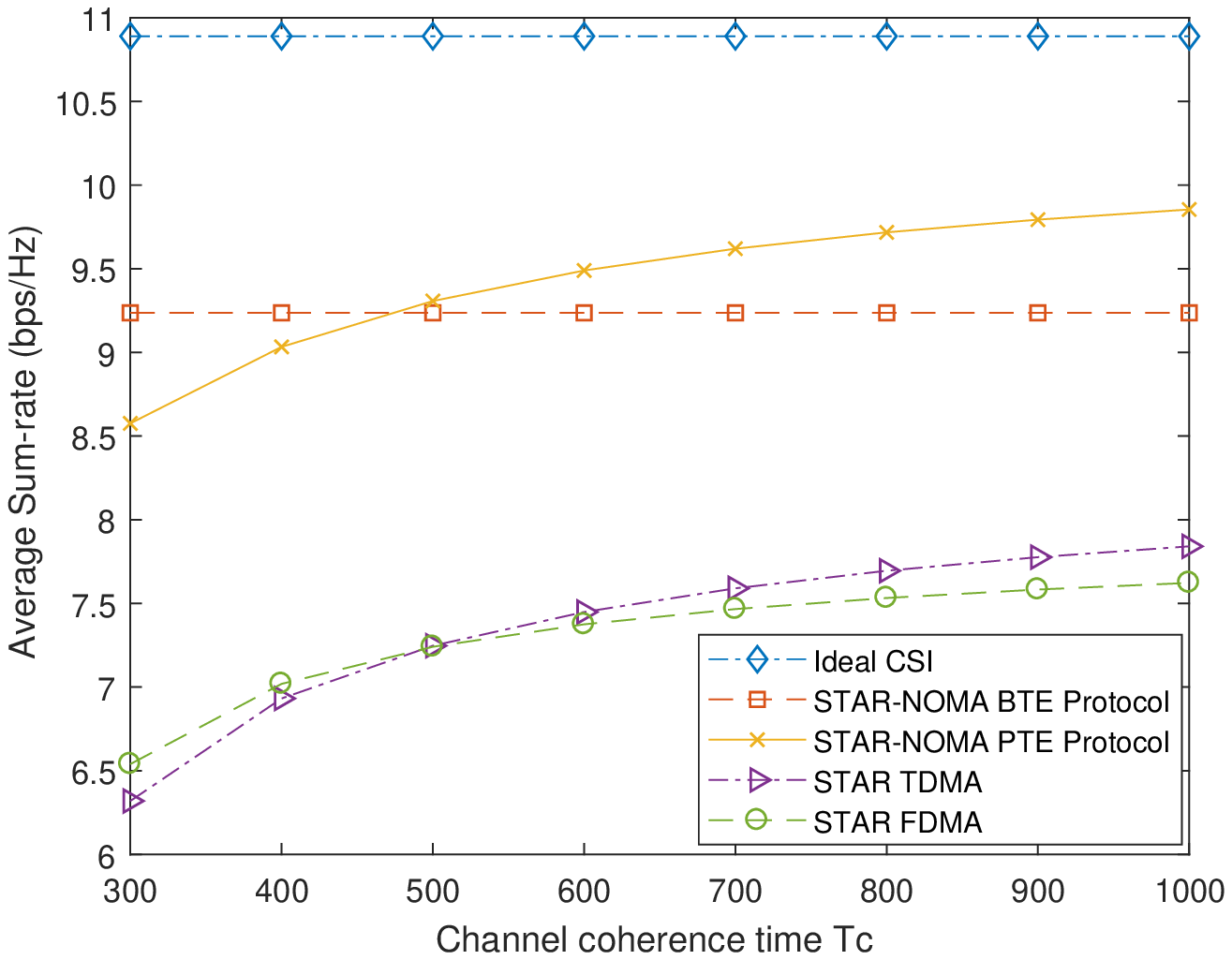}
		\caption{Average achievable sum-rate versus the channel coherence time $T_c$.}
		\label{rate_coherence}
	\end{minipage}
\end{figure}

We first consider a simple scenario with two users to show the superiority of NOMA over OMA.  In Fig. \ref{rate_coherence}, we plot the sum-rate of the users versus the channel coherence time $T_c$ with $M=50,\kappa=3, \gamma_k=3$ bps/Hz. Several important observations are made as follows. Firstly, it is observed that the performance of the proposed BTE protocol is not affected by $T_c$, since by exploiting S-CSI, almost all the time is utilized for data transmission. In contrast, the achieved sum-rate of the proposed PTE protocol increases with $T_c$ and is larger than that of the BTE protocol when $T_c$ is sufficiently large, i.e., $T_c\geq500$. Moreover, as $T_c$ grows, the impact of the channel estimation overhead on the system performance weakens. In this case, utilizing the I-CSI for the STAR-RIS beamforming design becomes more effective. Secondly, the STAR-RIS assisted systems based on NOMA largely outperform those based on FDMA and TDMA, since NOMA provides a larger multiplexing gain than FDMA/TDMA. This also indicates that NOMA is more suitable in practical systems for supporting more users. Finally, although the TDMA scheme can dynamically adjust STAR-RIS phase-shifts over time, it incurs a larger channel estimation overhead for each user, which is not practically effective.

Next, we consider the more general case with four users to show the benefits of deploying STAR-RISs. In Fig. \ref{fig_M}, we depict the sum-rate brought by different transmission protocols versus the number of total elements $M$. It is observed that the ideal CSI scheme offers the best performance. The performance gap between the ideal CSI case and the proposed BTE protocol is because the S-CSI can not well characterize the variations of the instantaneous channels, whereas for the proposed PTE protocol, the gap is mainly attributed to the training time of channel estimation. One can observe that there exists a tradeoff between the total number of elements and the system sum-rate of the PTE protocol. Specifically, the average rate under the PTE protocol first increases and then decreases with $M$. This can be explained as follows. When $M$ is small, a larger surface achieves a higher passive beamforming gain. However, as $M$ becomes larger, the penalty due to channel estimation overhead starts to dominate, which becomes the main performance bottleneck when $M$ exceeds a threshold (i.e., $M>50$). On the other hand, the sum-rate of the BTE protocol monotonically increases with $M$, because with S-CSI, the overhead for estimating the instantaneous effective channels in each coherence block is very small. Furthermore, as observed, the considered STAR-RIS aided systems significantly outperform the conventional RIS-aided systems for different $M$, and the performance gain increases with $M$. In a nutshell, the two proposed transmission protocols for STAR-RIS are generally preferred for the scenarios with a large and a small number of elements, respectively.

\begin{figure}[t!]
	\centering
	\begin{minipage}[t]{0.45\linewidth}
		\includegraphics[width=3 in]{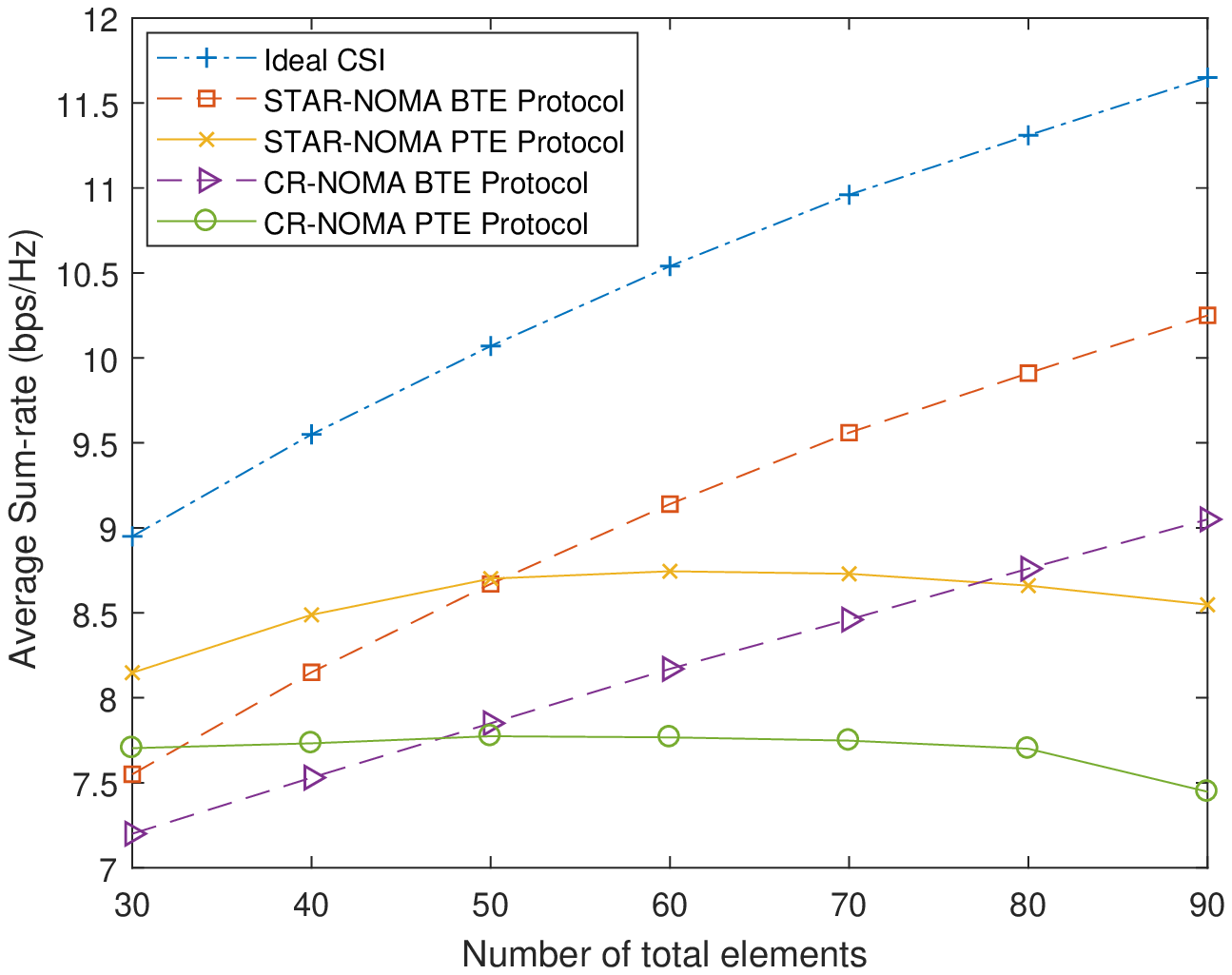}
		\caption{Average achievable sum-rate versus the number of total elements $M$ with $T_c=400, \kappa=3, \gamma_k=1.8$ bps/Hz.}
		\label{fig_M}
	\end{minipage}
	\quad
	\begin{minipage}[t]{0.45\linewidth}
		\includegraphics[width=3in]{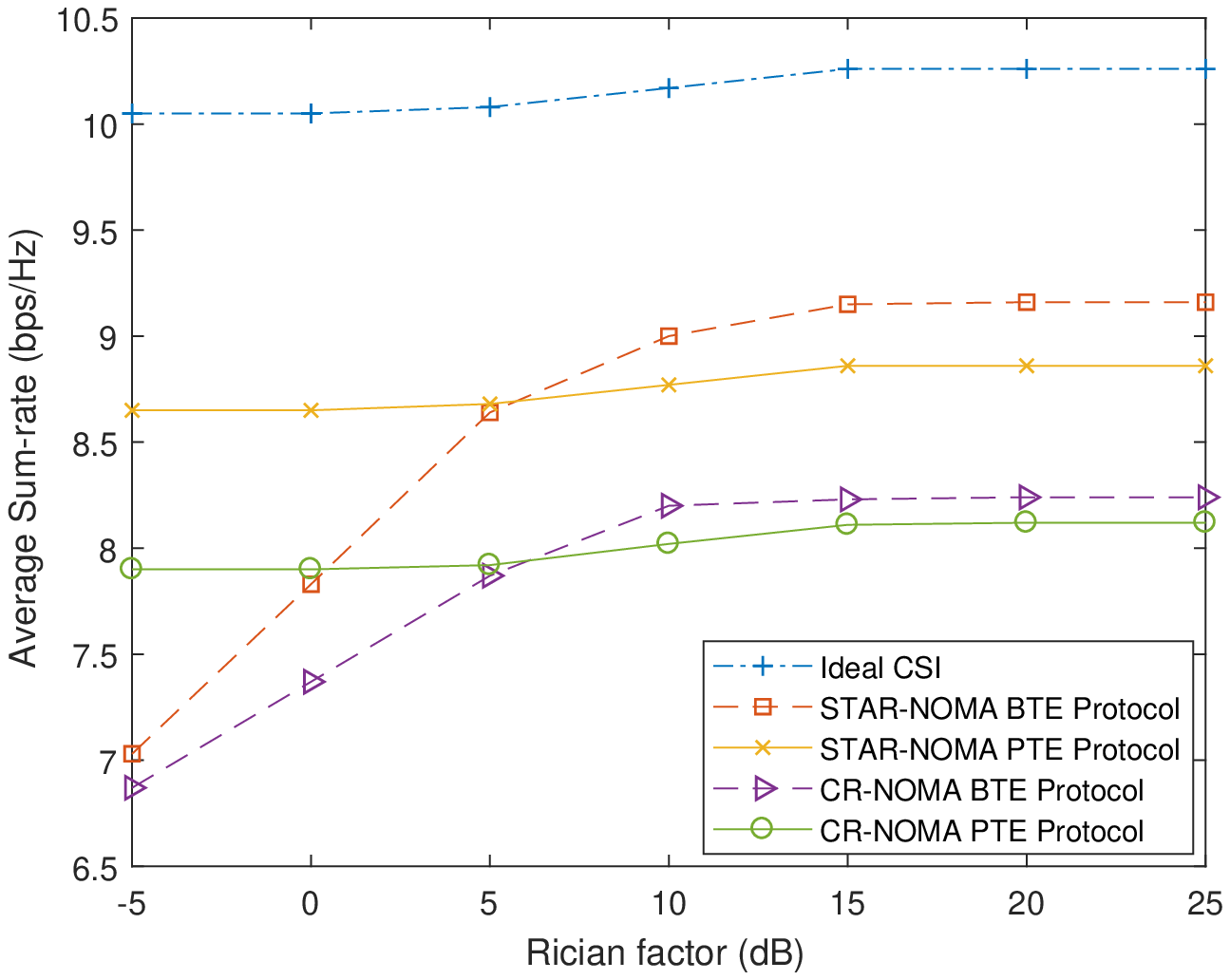}
		\caption{Average achievable sum-rate versus Rician factor $\kappa$ with $T_c=400, M=50, \gamma_k=1.8$ bps/Hz.}
		\label{fig_rician}
	\end{minipage}
\end{figure}

Fig. \ref{fig_rician} shows the rate performance of the proposed protocols versus the Rician factor $\kappa$. It is observed that both protocols for STAR-RIS-NOMA systems outperform their counterparts for CR-NOMA systems, since the STAR-RIS provides more flexibility in controlling the number of transmission/reflection elements, thus achieving a larger NOMA gain. 
Moreover, the performance of the BTE protocol increases drastically at first and then tends to be saturated. This is because the S-CSI can well characterize the actual channel condition in the high-Rician-factor regime, but may be ineffective when the NLoS components are dominant, as discussed in \textbf{Remark 1}. On the other hand, the sum-rate achieved by the PTE protocol slightly increases with $\kappa$, which indicates that the proposed PTE protocol is not sensitive to the fading condition. To summarize, the proposed BTE protocol outperforms the PTE protocol when the LoS components of the channel coefficients are dominant, and vice versa.

\vspace{-1em}

\section{Conclusions}
In this paper, we proposed two efficient TTS transmission protocols for different channel setups in a STAR-RIS aided NOMA communication system under MS. Specifically, for the LoS-dominant channel setup, a BTE protocol was first proposed, where the long-term STAR-RIS transmission- and reflection-coefficient vectors are optimized based on the S-CSI, while the short-term BS power allocation in each coherence block is designed based on the estimated effective CSI.  Moreover, a penalty-based method was proposed to optimize the STAR-RIS beamforming coefficients. Next, for the rich scattering channel setup, we further proposed a customized PTE protocol, where the long-term operation modes of STAR-RIS elements are optimized based on the information of large-scale path-loss, while the short-term BS power allocation and STAR-RIS phase-shifts are designed to cater to the estimated partial I-CSI. Moreover, a low-complexity surface-partition algorithm was devised to obtain the elements allocation assigned to serve each user. Numerical results validated the effectiveness of our proposed designs. Particularly, both transmission protocols can greatly reduce the channel estimation overhead while achieving high communication rates. 
\vspace{-1em}
\section*{Appendix A: Proof of Proposition 1}
Given the effective channel power gains of all users, the short-term optimization problem is equivalent to a power allocation problem in a single-cluster NOMA system. Following the method in\cite{noma_oma}, to maximize the system sum-rate under users' target rates, each user should be first allocated the minimum power for meeting its rate requirement, while the remaining power should be allocated to the user with the highest decoding order. Hence, the optimal power allocation is obtained by solving the following equations:
\begin{equation}\label{equation_power}
	\left\lbrace 
	\begin{aligned}
		&p_2=(2^{\gamma_2}-1)\left(p_1+\frac{\sigma^2}{{|c_2(\mathbf{v}_s)|^2}}\right),\\
		&\vdots\\
		&p_K=(2^{\gamma_K}-1)\left(\sum_{k=1}^{K-1}p_k+\frac{\sigma^2}{{|c_K(\mathbf{v}_s)|^2}}\right),\\
		&\sum_{k=1}^K{p_k}={P_{\text{max}}}.
	\end{aligned}
	\right.
\end{equation}

By solving the above equations, we obtain the optimal BS power allocation given in \eqref{optimal_power}.  Moreover, the original problem is feasible if all the power allocation factors $p_k^*$ are positive values and the strongest users' rate requirements are satisfied.
\vspace{-1em}
\section*{Appendix B: Proof of Lemma 1}
The expected effective channel power gains can be decomposed as follows
\begin{equation}\label{proof_exp}
	\mathbb{E}\left\lbrace |{h}_{k}+\mathbf{r}_{k}^H\boldsymbol{\Theta}_s\mathbf{g}|^2\right\rbrace \overset{(a)}{=}\mathbb{E}\{|h_k|^2\}+|x_1|^2+\mathbb{E}\{|x_2|^2\}+\mathbb{E}\{|x_3|^2\}+\mathbb{E}\{|x_4|^2\},
\end{equation}
where $x_1=\sqrt {\frac{\kappa_1\kappa_2{\delta_{BS}\delta_{Sk}}}{{({\kappa_1} + 1)(\kappa_2+1)}}} \bar{\mathbf{r}}_{k}^H\boldsymbol{\Theta}_s\bar{\mathbf{g}}$, $x_2=\sqrt {\frac{\kappa_1{\delta_{BS}\delta_{Sk}}}{{({\kappa_1} + 1)(\kappa_2+1)}}} \bar{\mathbf{r}}_{k}^H\boldsymbol{\Theta}_s\hat{\mathbf{g}}$, $x_3=\sqrt {\frac{\kappa_2{\delta_{BS}\delta_{Sk}}}{{({\kappa_2} + 1)(\kappa_2+1)}}} \hat{\mathbf{r}}_{k}^H\boldsymbol{\Theta}_s\bar{\mathbf{g}}$, $x_4=\sqrt {\frac{\delta_{BS}\delta_{Sk}}{{({\kappa_2} + 1)(\kappa_2+1)}}} \hat{\mathbf{r}}_{k}^H\boldsymbol{\Theta}_s\hat{\mathbf{g}}$, and $(a)$ is due to the fact that the small-scale fading channel coefficients $h_k, \hat{\mathbf{g}}, \hat{\mathbf{r}}_k$ are independent. Next, we have
\begin{equation}
|x_1|^2=\frac{\kappa_1\kappa_2{\delta_{BS}\delta_{Sk}}}{{({\kappa_1} + 1)(\kappa_2+1)}}\left|  {\bar{\mathbf{w}}_k^H\mathbf{v}_s}\right|^2,
\end{equation}
\begin{equation}
	\mathbb{E}\{|x_2|^2\}=\frac{\kappa_1{\delta_{BS}}}{{({\kappa_1} + 1)(\kappa_2+1)}}\bar{\mathbf{r}}_{k}^H\boldsymbol{\Theta}_s\mathbb{E}\{\hat{\mathbf{g}}\hat{\mathbf{g}}^H\}\boldsymbol{\Theta}_s\bar{\mathbf{r}}_{k}\overset{(b)}{=}\frac{\kappa_1{\delta_{BS}}}{{({\kappa_1} + 1)(\kappa_2+1)}}\sum_{m=1}^{M}\beta_m^s,
\end{equation}
where $(b)$ holds since $\mathbb{E}\{\hat{\mathbf{g}}\hat{\mathbf{g}}^H\}=\mathbf{I}_M$, $\boldsymbol{\Theta}_s\boldsymbol{\Theta}_s^H=\text{diag}(\boldsymbol{\beta}_s)$. The remaining terms can be obtained similarly as 
\begin{align}
	\mathbb{E}\{|x_3|^2\}=\frac{\kappa_2{\delta_{Sk}}}{{({\kappa_1} + 1)(\kappa_2+1)}}\sum_{m=1}^{M}\beta_m^s,\\
	\mathbb{E}\{|x_4|^2\}=\frac{1}{{({\kappa_1} + 1)(\kappa_2+1)}}\sum_{m=1}^{M}\beta_m^s.
\end{align}
Additionally, we have $\mathbb{E}\{|h_k|^2\}=\delta_k$. Therefore, by substituting (29)-(32)  into (\ref{proof_exp}), we arrive at (\ref{expect_channel}).

\vspace{-1em}
\section*{Appendix C: Proof of Lemma 3}
Applying the binomial expansion theorem, we have
\begin{equation}\label{lemma2_1}
	\begin{aligned}
		&\mathbb{E}\{(c_k)^2\}=\mathbb{E}\left\lbrace \left( |h_{k}|+\sum_{m=1}^{M_k}|[\mathbf{q}_{k}]_m|\right) ^2\right\rbrace =\\	 &\underbrace{\mathbb{E}\{|h_{k}|^2}_{x_1}\}+\underbrace{\mathbb{E}\left\lbrace \left(  \sum_{m=1}^{M_k}|[\mathbf{q}_{k}]_m|\right)^2 \right\rbrace }_{x_2} 
		+2\underbrace{\mathbb{E}\left\lbrace\left(  \sum_{m=1}^{M_k}|h_{k}||[\mathbf{q}_k]_m|\right) \right\rbrace}_{x_3},	
	\end{aligned}
\end{equation}
where the second term can be further expanded as
\begin{equation}\label{expe_2}
	\begin{aligned}
		x_2=&\mathbb{E}\left\lbrace   \sum_{m=1}^{M}|[\mathbf{q}_{k}]_m|^2 \right\rbrace + 
		\mathbb{E}\left\lbrace   \sum_{m=1}^{M_k}\sum_{n=1,n\neq{m}}^{M_k}|[\mathbf{q}_{k}]_m||[\mathbf{q}_{k}]_n| \right\rbrace. 
	\end{aligned}
\end{equation}

First, it is easy to have $\mathbb{E}\{|h_{k}|^2\}=\delta_k$, $\mathbb{E}\{|[\mathbf{q}_k]_m|^2\}=\delta_{BS}\delta_{Sk}$. 
Then, we note that $h_k$, $\mathbf{q}_k$ are independent variables with the following results
\begin{subequations}\label{distribution}
	\begin{align}
		&\mathbb{E}\{|h_{k}|\}=\frac{\sqrt{\pi{\delta_{k}}}}{2},\\
		\label{lemma2_2}&\mathbb{E}\{|[\mathbf{q}_k]_m|\}=\frac{\pi{\delta_{BS}\delta_{Sk}}}{4(\kappa_{1}+1)(\kappa_{2}+1)}L_{\frac{1}{2}}(-\kappa_{1})L_{\frac{1}{2}}(-\kappa_{2}), \forall {m} \in\mathcal{M}_k,
	\end{align}
\end{subequations}	
where \eqref{lemma2_2} is due to the fact that $[\mathbf{q}_k]_m$ is the product of two independent Rician variables, which follow the same distribution with the element in $\mathbf{g}$ and $\mathbf{r}_k$, respectively.

Substituting the above results into \eqref{lemma2_1} yields the desired result in \eqref{expected_channel_gain_scheme2}.

\bibliographystyle{IEEEtran}
\bibliography{mybib}

\begin{thebibliography}{10}
\begin{spacing}{1.25}
\providecommand{\url}[1]{#1}
\csname url@samestyle\endcsname
\providecommand{\newblock}{\relax}
\providecommand{\bibinfo}[2]{#2}
\providecommand{\BIBentrySTDinterwordspacing}{\spaceskip=0pt\relax}
\providecommand{\BIBentryALTinterwordstretchfactor}{4}
\providecommand{\BIBentryALTinterwordspacing}{\spaceskip=\fontdimen2\font plus
\BIBentryALTinterwordstretchfactor\fontdimen3\font minus
  \fontdimen4\font\relax}
\providecommand{\BIBforeignlanguage}[2]{{%
\expandafter\ifx\csname l@#1\endcsname\relax
\typeout{** WARNING: IEEEtran.bst: No hyphenation pattern has been}%
\typeout{** loaded for the language `#1'. Using the pattern for}%
\typeout{** the default language instead.}%
\else
\language=\csname l@#1\endcsname
\fi
#2}}
\providecommand{\BIBdecl}{\relax}
\BIBdecl

\bibitem{ris_survey2}
M.~Di~Renzo, A.~Zappone, M.~Debbah, M.-S. Alouini, C.~Yuen, J.~de~Rosny, and
  S.~Tretyakov, ``Smart radio environments empowered by reconfigurable
  intelligent surfaces: How it works, state of research, and the road ahead,''
  \emph{{IEEE} J. Sel. Areas Commun.}, vol.~38, no.~11, pp. 2450--2525, Nov.
  2020.

\bibitem{survey_ris}
Y.~Liu, X.~Liu, X.~Mu, T.~Hou, J.~Xu, M.~Di~Renzo, and N.~Al-Dhahir,
  ``Reconfigurable intelligent surfaces: Principles and opportunities,''
  \emph{{IEEE} Commun. Surv. Tut.}, vol.~23, no.~3, pp. 1546--1577,
  thirdquarter 2021.

\bibitem{zr_survey}
Q.~Wu, S.~Zhang, B.~Zheng, C.~You, and R.~Zhang, ``Intelligent reflecting
  surface-aided wireless communications: A tutorial,'' \emph{{IEEE} Trans.
  Commun.}, vol.~69, no.~5, pp. 3313--3351, May 2021.

\bibitem{irs_joint}
Q.~Wu and R.~Zhang, ``Intelligent reflecting surface enhanced wireless network
  via joint active and passive beamforming,'' \emph{{IEEE} Trans. Wireless
  Commun.}, vol.~18, no.~11, pp. 5394--5409, Nov. 2019.

\bibitem{mimoris}
C.~Pan, H.~Ren, K.~Wang, W.~Xu, M.~Elkashlan, A.~Nallanathan, and L.~Hanzo,
  ``Multicell {MIMO} communications relying on intelligent reflecting
  surfaces,'' \emph{{IEEE} Trans. Wireless Commun.}, vol.~19, no.~8, pp.
  5218--5233, Aug. 2020.

\bibitem{STAR}
Y.~Liu, X.~Mu, J.~Xu, R.~Schober, Y.~Hao, H.~V. Poor, and L.~Hanzo, ``{STAR}:
  Simultaneous transmission and reflection for 360° coverage by intelligent
  surfaces,'' \emph{{IEEE} Wireless Commun.}, vol.~28, no.~6, pp. 102--109,
  2021.

\bibitem{ios}
H.~Zhang, S.~Zeng, B.~Di, Y.~Tan, M.~{Di Renzo}, M.~Debbah, Z.~Han, H.~V. Poor,
  and L.~Song, ``Intelligent omni-surfaces for full-dimensional wireless
  communications: Principles, technology, and implementation,'' \emph{{IEEE}
  Commun. Mag.}, vol.~60, no.~2, pp. 39--45, Feb. 2022.

\bibitem{proceeding}
Y.~{Liu}, Z.~{Qin}, M.~{Elkashlan}, Z.~{Ding}, A.~{Nallanathan}, and
  L.~{Hanzo}, ``Nonorthogonal multiple access for 5{G} and beyond,''
  \emph{Proc. {IEEE}}, vol. 105, no.~12, pp. 2347--2381, 2017.

\bibitem{power_noma2}
O.~Maraqa, A.~S. Rajasekaran, S.~Al-Ahmadi, H.~Yanikomeroglu, and S.~M. Sait,
  ``A survey of rate-optimal power domain {NOMA} with enabling technologies of
  future wireless networks,'' \emph{{IEEE} Commun. Surv. Tut.}, vol.~22, no.~4,
  pp. 2192--2235, fourthquarter 2020.

\bibitem{ding_noma_tvt}
Z.~Ding, P.~Fan, and H.~V. Poor, ``Impact of user pairing on {5G} nonorthogonal
  multiple-access downlink transmissions,'' \emph{{IEEE} Trans. Veh. Technol.},
  vol.~65, no.~8, pp. 6010--6023, Aug. 2016.

\bibitem{ofdm}
Y.~Yang, B.~Zheng, S.~Zhang, and R.~Zhang, ``Intelligent reflecting surface
  meets {OFDM}: Protocol design and rate maximization,'' \emph{{IEEE} Trans.
  Commun.}, vol.~68, no.~7, pp. 4522--4535, July 2020.

\bibitem{irs_mmwave}
P.~Wang, J.~Fang, X.~Yuan, Z.~Chen, and H.~Li, ``Intelligent reflecting
  surface-assisted millimeter wave communications: Joint active and passive
  precoding design,'' \emph{{IEEE} Trans. Veh. Technol.}, vol.~69, no.~12, pp.
  14\,960--14\,973, Dec. 2020.

\bibitem{huameng_uav}
M.~Hua, L.~Yang, Q.~Wu, and A.~L. Swindlehurst, ``{3D UAV} trajectory and
  communication design for simultaneous uplink and downlink transmission,''
  \emph{{IEEE} Trans. Commun.}, vol.~68, no.~9, pp. 5908--5923, Sept. 2020.

\bibitem{risnoma4}
M.~Zeng, X.~Li, G.~Li, W.~Hao, and O.~A. Dobre, ``Sum rate maximization for
  {IRS}-assisted uplink {NOMA},'' \emph{{IEEE} Commun. Lett.}, vol.~25, no.~1,
  pp. 234--238, Jan. 2021.

\bibitem{risnoma7}
G.~Yang, X.~Xu, Y.-C. Liang, and M.~Di~Renzo, ``Reconfigurable intelligent
  surface-assisted non-orthogonal multiple access,'' \emph{{IEEE} Trans.
  Wireless Commun.}, vol.~20, no.~5, pp. 3137--3151, May 2021.

\bibitem{risnoma8}
M.~Fu, Y.~Zhou, Y.~Shi, and K.~B. Letaief, ``Reconfigurable intelligent surface
  empowered downlink non-orthogonal multiple access,'' \emph{{IEEE} Trans.
  Commun.}, vol.~69, no.~6, pp. 3802--3817, June 2021.

\bibitem{dft}
T.~L. Jensen and E.~De~Carvalho, ``An optimal channel estimation scheme for
  intelligent reflecting surfaces based on a minimum variance unbiased
  estimator,'' in \emph{Proc. IEEE Int. Conf. Acoust., Speech Signal Process.
  (ICASSP)}, May 2020, pp. 5000--5004.

\bibitem{discrete_jsac}
C.~You, B.~Zheng, and R.~Zhang, ``Channel estimation and passive beamforming
  for intelligent reflecting surface: Discrete phase shift and progressive
  refinement,'' \emph{{IEEE} J. Sel. Areas Commun.}, vol.~38, no.~11, pp.
  2604--2620, Nov. 2020.

\bibitem{ce_ofdma}
B.~Zheng, C.~You, and R.~Zhang, ``Intelligent reflecting surface assisted
  multi-user {OFDMA}: Channel estimation and training design,'' \emph{{IEEE}
  Trans. Wireless Commun.}, vol.~19, no.~12, pp. 8315--8329, Dec. 2020.

\bibitem{ll}
Z.~Wang, L.~Liu, and S.~Cui, ``Channel estimation for intelligent reflecting
  surface assisted multiuser communications: Framework, algorithms, and
  analysis,'' \emph{{IEEE} Trans. Wireless Commun.}, vol.~19, no.~10, pp.
  6607--6620, Oct. 2020.

\bibitem{part2}
X.~Wei, D.~Shen, and L.~Dai, ``Channel estimation for {RIS} assisted wireless
  communications—part {II}: An improved solution based on double-structured
  sparsity,'' \emph{{IEEE} Commun. Lett.}, vol.~25, no.~5, pp. 1403--1407, May
  2021.

\bibitem{beamtraining}
C.~You, B.~Zheng, and R.~Zhang, ``Fast beam training for {IRS}-assisted
  multiuser communications,'' \emph{{IEEE} Wireless Commun. Lett.}, vol.~9,
  no.~11, pp. 1845--1849, Nov. 2020.

\bibitem{scsi1}
Y.~Han, W.~Tang, S.~Jin, C.-K. Wen, and X.~Ma, ``Large intelligent
  surface-assisted wireless communication exploiting statistical {CSI},''
  \emph{{IEEE} Trans. Veh. Technol.}, vol.~68, no.~8, pp. 8238--8242, Aug.
  2019.

\bibitem{scsi2}
K.~Zhi, C.~Pan, H.~Ren, and K.~Wang, ``Power scaling law analysis and phase
  shift optimization of {RIS}-aided massive {MIMO} systems with statistical
  {CSI},'' \emph{{IEEE} Trans. Commun.}, vol.~70, no.~5, pp. 3558--3574, May
  2022.

\bibitem{downlink_scsi}
Q.~Tao, S.~Zhang, C.~Zhong, W.~Xu, H.~Lin, and Z.~Zhang, ``Weighted sum-rate of
  intelligent reflecting surface aided multiuser downlink transmission with
  statistical {CSI},'' \emph{{IEEE} Trans. Wireless Commun.}, pp. 1--1, early
  access, 2021, doi=10.1109/TWC.2021.3134822.

\bibitem{kangda2}
K.~Zhi, C.~Pan, H.~Ren, K.~Wang, M.~Elkashlan, M.~Di~Renzo, R.~Schober, H.~V.
  Poor, J.~Wang, and L.~Hanzo, ``Two-timescale design for reconfigurable
  intelligent surface-aided massive {MIMO} systems with imperfect {CSI},''
  [Online]. Available: \url{https://arxiv.org/abs/2108.07622}.

\bibitem{marco3}
A.~Abrardo, D.~Dardari, and M.~Di~Renzo, ``Intelligent reflecting surfaces:
  Sum-rate optimization based on statistical position information,''
  \emph{{IEEE} Trans. Commun.}, vol.~69, no.~10, pp. 7121--7136, Oct. 2021.

\bibitem{SRAR}
J.~Xu, Y.~Liu, X.~Mu, and O.~A. Dobre, ``{STAR-RIS}s: Simultaneous transmitting
  and reflecting reconfigurable intelligent surfaces,'' \emph{{IEEE} Commun.
  Lett.}, vol.~25, no.~9, pp. 3134--3138, Sept. 2021.

\bibitem{ios2}
S.~Zeng, H.~Zhang, B.~Di, Y.~Liu, M.~Di~Renzo, Z.~Han, H.~V. Poor, and L.~Song,
  ``Intelligent omni-surfaces: Reflection-refraction circuit model,
  full-dimensional beamforming, and system implementation,'' [Online].
  Available: \url{https://arxiv.org/abs/2206.00204}.

\bibitem{star1}
X.~Mu, Y.~Liu, L.~Guo, J.~Lin, and R.~Schober, ``Simultaneously transmitting
  and reflecting {(STAR) RIS} aided wireless communications,'' \emph{{IEEE}
  Trans. Wireless Commun.}, Early access, 2021, doi:{10.1109/TWC.2021.3118225}.

\bibitem{star2}
C.~Wu, Y.~Liu, X.~Mu, X.~Gu, and O.~A. Dobre, ``Coverage characterization of
  {STAR-RIS} networks: {NOMA and OMA},'' \emph{{IEEE} Commun. Lett.}, vol.~25,
  no.~9, pp. 3036--3040, Sept. 2021.

\bibitem{star4}
J.~Zuo, Y.~Liu, Z.~Ding, L.~Song, and H.~V. Poor, ``Joint design for
  simultaneously transmitting and reflecting {(STAR) RIS } assisted {NOMA}
  systems,'' [Online]. Available: \url{https://arxiv.org/abs/2106.03001}.

\bibitem{analysis_ios}
C.~Zhang, W.~Yi, Y.~Liu, Z.~Ding, and L.~Song, ``{STAR-IOS} aided {NOMA}
  networks: Channel model approximation and performance analysis,''
  \emph{{IEEE} Trans. Wireless Commun.}, pp. 1--1, 2022, early access,
  doi={10.1109/TWC.2022.3152703}.

\bibitem{ce_survey}
B.~Zheng, C.~You, W.~Mei, and R.~Zhang, ``A survey on channel estimation and
  practical passive beamforming design for intelligent reflecting surface aided
  wireless communications,'' \emph{{IEEE} Commun. Surv. Tut.}, pp. 1--1, 2022,
  early access, doi={10.1109/COMST.2022.3155305}.

\bibitem{star_ce}
C.~Wu, C.~You, Y.~Liu, X.~Gu, and Y.~Cai, ``Channel estimation for
  {STAR-RIS}-aided wireless communication,'' \emph{{IEEE} Commun. Lett.},
  vol.~26, no.~3, pp. 652--656, March 2022.

\bibitem{couple2}
J.~Xu, Y.~Liu, X.~Mu, R.~Schober, and H.~V. Poor, ``{STAR-RISs}: A correlated
  {T\&R} phase-shift model and practical phase-shift configuration
  strategies,'' \emph{{IEEE} J. Sel. Topics Signal Process.}, pp. 1--1, 2022,
  early access, doi=10.1109/JSTSP.2022.3175030.

\bibitem{mingmin}
M.-M. Zhao, Q.~Wu, M.-J. Zhao, and R.~Zhang, ``Intelligent reflecting surface
  enhanced wireless networks: Two-timescale beamforming optimization,''
  \emph{{IEEE} Trans. Wireless Commun.}, vol.~20, no.~1, pp. 2--17, Jan. 2021.

\bibitem{aoa}
G.~Zhou, C.~Pan, H.~Ren, P.~Popovski, and A.~L. Swindlehurst, ``Channel
  estimation for {RIS}-aided multiuser millimeter-wave systems,'' \emph{{IEEE}
  Trans. Signal Process}, vol.~70, pp. 1478--1492, 2022.

\bibitem{cvx}
M.~Grant and S.~Boyd, ``{CVX}: Matlab software for disciplined convex
  programming,'' [Online]. Available:\url{http://cvxr.com/cvx}, 2014.

\bibitem{convex}
S.~{B}oyd and L.~{V}andenberghe, \emph{Convex Optimization}.\hskip 1em plus
  0.5em minus 0.4em\relax Cambridge, U.K.: Cambridge Univ. Press, 2004.

\bibitem{noma_oma}
Z.~{Chen}, Z.~{Ding}, X.~{Dai}, and R.~{Zhang}, ``An optimization perspective
  of the superiority of {NOMA} compared to conventional {OMA},'' \emph{{IEEE}
  Trans. Signal Process}, vol.~65, no.~19, pp. 5191--5202, Oct. 2017.
\end{spacing}
\end{thebibliography}
\end{document}